\documentclass[smallextended]{svjour3}
\usepackage{mathptmx}
\usepackage{graphicx}
\usepackage{amsmath,amsfonts}
\usepackage{amssymb}
\DeclareMathOperator{\poly}{poly}

\smartqed
\usepackage{tikz}
\usetikzlibrary{quotes, backgrounds, positioning, calc, arrows.meta, decorations.pathreplacing, shapes}
\tikzset{every picture/.style={draw=black, line width=.5pt},font=\sffamily}
\tikzset{every node/.style={circle, inner sep=0pt, minimum size=14},font=\sffamily}
\tikzset{>={Latex[color=,length=6pt,width=2.5pt 3]},font=\sffamily}
\tikzstyle{graph} = [fill=black!3, fill opacity=1, draw=black,font=\sffamily]
\tikzstyle{source} = [fill=yellow!30, fill opacity=1, draw=black,font=\sffamily]
\tikzstyle{joint} = [fill=black!3, fill opacity=1, draw=black,font=\sffamily]
\tikzstyle{tree} = [fill=red!30, fill opacity=1, draw=black,font=\sffamily]
\tikzstyle{stain} = [fill=yellow!80!orange, draw=black, line width=1pt, fill opacity=.15, dotted]
\usepackage{algorithm}
\usepackage[noend]{algpseudocode}
\usepackage{booktabs}
\usepackage{microtype}
\usepackage{nicefrac}
\usepackage{thmtools}
\usepackage{thm-restate}

\begin{document}
\newcommand{\ch}{\lfloor\frac{k}{3}\rfloor}
\newcommand{\tH}{\lfloor\frac{k}{4}\rfloor}
\newcommand{\mysp}{8pt}
\newcommand{\sds}{\ensuremath{\mathcal{C}}}
\newcommand{\ds}{\ensuremath{\mathcal{P}}}
\newcommand{\ord}[1]{\operatorname{c}(\ensuremath{#1})}
\newcommand{\roots}{S}
\newcommand{\mkey}{J}
\newcommand{\clos}{\ensuremath{P}}
\newcommand{\cnt}[1]{\ensuremath{\operatorname{ind}(#1)}}
\newcommand{\sub}[1]{\ensuremath{\operatorname{sub}(#1)}}
\newcommand{\ind}[1]{\ensuremath{\operatorname{ind}(#1)}}
\newcommand{\indp}[2]{\ensuremath{\operatorname{ind}_{#2}(#1)}}
\newcommand{\cntp}[1]{\text{sub}\ensuremath{^{+}(#1)}}
\newcommand{\inj}[1]{\text{inj}\ensuremath{(#1)}}
\newcommand{\injA}[2]{\text{inj}\ensuremath{_{#1}(#2)}}
\newcommand{\aut}[1]{\text{aut}\ensuremath{(#1)}}
\newcommand{\homo}[1]{\ensuremath{\operatorname{hom}(#1)}}
\newcommand{\degr}[1]{\operatorname{deg}(#1)}
\newcommand{\sk}{\ensuremath{\Lambda}}
\newcommand{\core}{\ensuremath{C_i}}
\newcommand{\encg}{\ensuremath{\Gamma}}
\newcommand{\skel}[1]{\ensuremath{\Lambda}(#1)}
\newcommand{\skeli}[2]{\ensuremath{\Lambda_{#1}}(#2)}
\newcommand{\CountIndPlus}{CountInd$^{*}$}
\newcommand{\CountIndPlusBis}{CountInd}
\newcommand{\JointCount}{JointCount}
\newcommand{\injCond}[2]{\text{inj}\ensuremath{(#1\,|\,#2)}}
\newcommand{\homoCond}[2]{\text{hom}\ensuremath{(#1\,|\,#2)}}
\newcommand{\countDagInj}{\text{CountDagInjections}}
\newcommand{\enumerateChunk}{\text{EnumPiece}}
\newcommand{\enumDecomp}{\text{EnumDecomposition}}
\newcommand{\treeCount}{\text{HomCount}}
\newcommand{\subTreeCount}[2]{\ensuremath{\operatorname{c}_{#1}(#2)}}
\newcommand{\tw}{\operatorname{t}}
\newcommand{\sw}{\tau}
\newcommand{\hsw}{\tau}
\newcommand{\pluseq}{\mathrel{+}=}
\newcommand{\simpleCount}{\text{SimpleCount}}
\newcommand{\cov}[1]{C}
\newcommand{\buildThick}{\ensuremath{\textsc{Core}}}
\newcommand{\decompose}{\ensuremath{\textsc{Decompose}}}
\newcommand{\buildTree}{\ensuremath{\textsc{BuildTree}}}
\newcommand{\cleanSkel}{\ensuremath{\textsc{RRS}}}
\newcommand{\pth}{\!\leftrightsquigarrow\!}
\newcommand{\polylog}{\operatorname{polylog}}
\newcommand{\oline}[1]{\ensuremath{\Gamma[#1]}}
\newcommand{\bu}{\mathbf{u}}
\newcommand{\baru}{\bar{u}}
\newcommand{\barbu}{\bar{\mathbf{u}}}
\newcommand{\bv}{\mathbf{v}}
\newcommand{\barv}{\bar{v}}
\newcommand{\barbv}{\bar{\mathbf{v}}}
\newcommand{\vt}{\mathcal{B}}
\newcommand{\et}{\mathcal{E}}
\newcommand{\bul}{\bullet}
\newcommand{\skic}{\sk_i^{\bul}}
\newcommand{\mkeyc}{\mkey^{\bul}}
\newcommand{\rootsc}{\roots^{\,\bul}}
\newcommand{\edgc}{E^{\bul}}
\newcommand{\Tc}{T^{\,\bul}}
\newcommand{\vtc}{\vt^{\,\bul}}
\newcommand{\etc}{\et^{\,\bul}}
\newcommand{\Pic}{P_i^{\,\bul}}

\title{Faster algorithms for counting subgraphs in sparse graphs\thanks{A preliminary version of this article appeared in the proceedings of the 14th International Symposium on Parameterized and Exact Computation (IPEC 2019).}}
\author{Marco Bressan\thanks{Marco Bressan is supported in part by a Google Focused Award ``Algorithms and Learning for AI'' (ALL4AI), by the ERC Starting Grant DMAP 680153, and by the ``Dipartimenti di Eccellenza 2018-2022'' grant awarded to the Department of Computer Science of the Sapienza University of Rome.}}
\institute{Dipartimento di Informatica, Sapienza Universit\`a di Roma, Italy. \email{bressan@di.uniroma1.it}}
\maketitle

\begin{abstract}
Given a $k$-node pattern graph $H$ and an $n$-node host graph $G$, the subgraph counting problem asks to compute the number of copies of $H$ in $G$.
In this work we address the following question: can we count the copies of $H$ faster if $G$ is sparse?
We answer in the affirmative by introducing a novel tree-like decomposition for directed acyclic graphs, inspired by the classic tree decomposition for undirected graphs.
This decomposition gives a dynamic program for counting the homomorphisms of $H$ in $G$ by exploiting the degeneracy of $G$, which allows us to beat the state-of-the-art subgraph counting algorithms when $G$ is sparse enough.
For example, we can count the induced copies of any $k$-node pattern $H$ in time $2^{O(k^2)} O(n^{0.25k + 2} \log n)$ if $G$ has bounded degeneracy, and in time $2^{O(k^2)} O(n^{0.625k + 1} \log n)$ if $G$ has bounded average degree.
These bounds are instantiations of a more general result, parameterized by the degeneracy of $G$ and the structure of $H$, which generalizes classic bounds on counting cliques and complete bipartite graphs.
We also give lower bounds based on the Exponential Time Hypothesis, showing that our results are actually a characterization of the complexity of subgraph counting in bounded-degeneracy graphs.

\CRclass{Mathematics of computing -- Discrete mathematics \and Theory of computation -- Design and analysis of algorithms \and Theory of computation -- Graph algorithms analysis}
\end{abstract}

\section{Introduction}
\label{sec:intro}
We address the following fundamental subgraph counting problem:
\begin{quote}
\textbf{Input:} an $n$-node graph $G$ (the \emph{host graph}) and a $k$-node graph $H$ (the \emph{pattern})\\
\textbf{Output:} the number of induced copies of $H$ in $G$
\end{quote}
If no further assumptions are made, the best possible algorithm for this problem is likely to have running time $f(k) \cdot n^{\Theta(k)}$.
Indeed, the naive brute-force algorithm has running time $O(k^2 n^k)$, and under the Exponential Time Hypothesis~\cite{Impagliazzo&1998} any algorithm for counting $k$-cliques has running time $n^{\Omega(k)}$~\cite{Chen&2005,Chen&2006}.
The best algorithm known, which was given over $30$ years ago by Ne\v{s}et\v{r}il and Poljak~\cite{Nesetril&1985} and is based on fast matrix multiplication, is only slightly faster than $O(k^2n^k)$.
Ignoring $\poly(k)$ factors\footnote{In this paper we suppress $\poly(k)$ factors by default; if needed, we explicit them in order to emphasize that the dependence on $k$ is polynomial rather than exponential.}, the algorithm runs in time $O(n^{\omega\lfloor\frac{k}{3}\rfloor + 2})$ where $\omega$ is the matrix multiplication exponent.
Since $\omega \le 2.373$~\cite{LeGall2014}, this gives a state-of-the-art running time of $O( n^{0.791 k + 2})$.

In this work, we aim at breaking through this ``$n^{\Theta(k)}$ barrier'' by assuming that $G$ is sparse, and in particular, that $G$ has bounded degeneracy.
This assumption is often made for real-world graphs like social networks, since it agrees well with their structural properties~\cite{Eppstein&2011}.
The family of bounded-degeneracy graphs is rich from a theoretical point of view, too: it includes many important classes such as Barab\'asi-Albert preferential attachment graphs, graphs excluding a fixed minor, planar graphs, bounded-treewidth graphs, bounded-degree graphs, and bounded-genus graphs, see~\cite{Grohe&2013nowheredense}.
Unfortunately, even when $G$ has bounded degeneracy, the state of the art remains the $O(n^{0.791 k + 2})$-time algorithm by Ne\v{s}et\v{r}il and Poljak, unless one makes further assumptions.
For example, one can count the copies of any given pattern $H$ in time $O(n)$, provided $G$ is planar~\cite{Eppstein1995} or has bounded treewidth~\cite{Nesetril2012sparsity} or has bounded degree~\cite{Patel&2018}; all conditions that are stricter than bounded degeneracy.
Alternatively, if $G$ has bounded degeneracy, $O(n)$-time algorithms exist when $H$ is the clique~\cite{Alon&2008,Chiba&1985,Eppstein&2010}, or when $H$ is a complete bipartite graph, if we do not require the copies of $H$ to be induced~\cite{Eppstein1994}.
Unfortunately, it is not clear how to extend the techniques behind these results to all patterns $H$ and all $G$ with bounded degeneracy.
Thus, to what extent a small degeneracy of $G$ makes subgraph counting easier remains an open question.

In this work we introduce a novel tree-like graph decomposition, to be applied to the pattern graph $H$, designed to exploit the degeneracy of $G$ when counting the homomorphisms of $H$ in $G$.
When $G$ is sparse enough, this decomposition yields subgraph counting algorithms faster than the state of the art.
For example, we show how to count the induced copies of \emph{any} $k$-node pattern $H$ in time $2^{O(k^2)} O(n^{0.25 k + 2} \log n)$ when $G$ has bounded degeneracy, and in time $2^{O(k^2)} O(n^{0.625 k + 2} \log n)$ when $G$ has bounded average degree.
These results are instantiations of a more general result which says that $H$ can be counted in time $f(k) O(d^{k-\hsw(H)} n^{\hsw(H)} \log n)$, where $d$ is the degeneracy of $G$, and $\hsw(H)$ is a certain measure of ``width'' of $H$ arising from our decomposition.
Assuming the Exponential Time Hypothesis, we also show that $n^{\Omega(\hsw(H)/\log \hsw(H))}$ operations are required in the worst case, even if $G$ has degeneracy $2$.
This provides a novel characterization of the complexity of subgraph counting in bounded-degeneracy graphs.

\subsection{Results}
\label{sub:results}
We divide our results into \emph{bounds} (Section~\ref{sub:bounds}) and \emph{techniques} (Section~\ref{sub:tech}).
We denote by $d$ the degeneracy of $G$, and we denote by \homo{H,G}, \sub{H,G}, \ind{H,G} the number of, respectively, homomorphisms, occurrences, and induced occurrences of $H$ in $G$. 
See Subsection~\ref{sub:prelim} for further definitions and notation. 
We remark that, unless otherwise specified, our bounds hold for every $H$ including disconnected ones.

\subsubsection{Bounds}
\label{sub:bounds}
Our first results are two running time bounds parameterized by the sparsity of $G$. 
\begin{theorem}
\label{thm:mainbound}
For any $k$-node pattern $H$ one can compute \homo{H,G} and \sub{H,G}  in time $2^{O(k \log k)} \cdot O(d^{k-(\tH+2)} n^{\tH+2} \log n)$, and one can compute \ind{H,G} in time $2^{O(k^2)} \cdot O(d^{k-(\tH+2)} n^{\tH+2} \log n)$, where $d$ is the degeneracy of $G$. 
\end{theorem}
This bound reduces the exponent of $n$ to $\tH+2 \le 0.25 k + 2$, down from the state-of-the-art $\omega \lfloor \frac{k}{3} \rfloor + 2 \le 0.791 k + 2$ of the Ne\v{s}et\v{r}il-Poljak bound.
This implies that our polynomial dependence on $n$ is better whenever $d = O(n^{0.721})$, and in any case (that is, even if $\omega=2$) whenever $d = O(n^{0.556})$.
As a corollary of Theorem~\ref{thm:mainbound}, since $d = O(\sqrt{rn})$ where $r$ is the average degree of $G$, we obtain:
\begin{theorem}
\label{thm:lowavgd}
For any $k$-node pattern $H$ one can compute \homo{H,G} and \sub{H,G} in time $2^{O(k \log k)} \cdot O(r^{\frac{1}{2}(k-\tH)-1} n^{\frac{1}{2}(k+\tH)+1} \log n)$, and one can compute \ind{H,G} in time $2^{O(k^2)} \cdot O(r^{\frac{1}{2}(k-\tH)-1} n^{\frac{1}{2}(k+\tH)+1} \log n)$, where $r$ is the average degree of $G$. 
\end{theorem}
This bound has a polynomial dependence on $n$ better than Ne\v{s}et\v{r}il-Poljak whenever $r = O(n^{0.221})$, and in any case (that is, even if $\omega=2$) whenever $r=O(n^{0.056})$.
In particular, we have a $2^{O(k^2)} \cdot O(n^{0.625k+1} \log n)$-time algorithm when $r=O(1)$.
These are the first improvements over the Ne\v{s}et\v{r}il-Poljak algorithm for graphs with small degeneracy or small average degree.

As a second result, we give improved bounds for some classes of patterns.
The first is the class of quasi-cliques, a typical target pattern for social networks~\cite{Sariyuce&2018,Sariyuce&2017,Tsourakakis&2017}.
We prove:
\begin{theorem}
\label{thm:dense_ub}
If $H$ is the clique minus $\epsilon$ edges, then one can compute \homo{H,G} and \sub{H,G} in time $2^{O(k \log k)} \cdot O(d^{k-\lceil \frac{1}{2} + \sqrt{\frac{\epsilon}{2}} \, \rceil} n^{\lceil \frac{1}{2} + \sqrt{\frac{\epsilon}{2}} \, \rceil} \log n)$, and \ind{H,G} in time $2^{O(\epsilon + k \log k)} \cdot O(d^{k-\lceil \frac{1}{2} + \sqrt{\frac{\epsilon}{2}} \, \rceil} n^{\lceil \frac{1}{2} + \sqrt{\frac{\epsilon}{2}} \, \rceil} \log n)$.
\end{theorem}
\noindent 
This generalizes the classic  $O(d^{k-1} n)$ bound for counting cliques by Chiba and Nishizeki~\cite{Chiba&1985}, at the price of an extra factor $2^{O(\epsilon + k \log k)} O(\log n)$.
Next, we consider complete quasi-multipartite graphs:
\begin{theorem}
\label{thm:bipartite}
If $H$ is a complete multipartite graph, then one can compute $\homo{H,G}$ and $\sub{H,G}$ in time $2^{O(k \log k)} \cdot O(d^{k-1} n \log n)$.
If $H$ is a complete multipartite graph plus $\epsilon$ edges, then one can compute $\homo{H,G}$ and $\sub{H,G}$ in time $2^{O(k \log k)} \cdot O(d^{k-\lfloor\frac{\epsilon}{4}\rfloor-2} n^{\lfloor\frac{\epsilon}{4}\rfloor+2} \log n)$.
\end{theorem}
\noindent 
This generalizes an existing $O(d^3 2^{2d} n)$ bound for counting the non-induced copies of complete (maximal) bi-partite graphs~\cite{Eppstein1994}, again at the price of an extra factor $2^{O(k \log k)} \log n$.

Table~\ref{tab:ub} summarizes our upper bounds.
We remark that our algorithms work for the colored versions of the problem (count only copies of $H$ with prescribed vertex and/or edge colors) as well as the weighted versions of the problem (compute the total node or edge weight of copies of $H$ in $G$).
This can be obtained by a straightforward adaptation of our homomorphism counting algorithms.

\renewcommand{\arraystretch}{1.2}
\begin{table}[h]
\resizebox{.99\textwidth}{!}{
\centering
\begin{tabular}{lll}
pattern $H$ & time to compute \ind{H,G} & reference \\
\toprule
all (even disconnected) & $O\big(n^{\omega\lfloor\frac{k}{3}\rfloor + 2}\big)$ & \cite{Nesetril&1985}\\
all (even disconnected) & $2^{O(k^2)} \cdot O\big(d^{k-\tH-2} n^{\tH+2} \log n\big)$ & this work\\
all (even disconnected) & $2^{O(k^2)} \cdot O\big(r^{\frac{1}{2}(k-\tH)-1} n^{\frac{1}{2}(k+\tH)+1} \log n\big)$ & this work \\
$K_k$ & $O\big(d^{k-1} n\big)$  & \cite{Chiba&1985} \\
$K_k$ - $\epsilon$ edges & $2^{O(\epsilon + k \log k)} \cdot O\big(d^{k-\lceil \frac{1}{2} + \sqrt{\frac{\epsilon}{2}} \, \rceil} n^{\lceil \frac{1}{2} + \sqrt{\frac{\epsilon}{2}} \, \rceil} \log n\big)$ & this work  \\
$K_{k_1,k_2}$ & $O(d^3 2^{2d} n)$  & \cite{Eppstein1994} \\
$K_{k_1,\ldots,k_{\ell}}$ & $2^{O(k \log k)} \cdot O(d^{k-1} n \log n)$ \;\, {\quad\quad\quad\quad (\sub{H,G} only}) & this work  \\
$K_{k_1,\ldots,k_{\ell}}$ + $\epsilon$ edges & $2^{O(k \log k)} \cdot O(d^{k-\lfloor\frac{\epsilon}{4}\rfloor-2} n^{\lfloor\frac{\epsilon}{4}\rfloor+2} \log n)$  \; (\sub{H,G} only) & this work
\\\bottomrule
\end{tabular}
}
\caption{Summary of upper bounds.}
\label{tab:ub}
\end{table}

\subsubsection{Techniques}
\label{sub:tech}
The bounds of Subsection~\ref{sub:bounds} are instantiations of a single, more general result.
This result is based on a novel notion of width, the \emph{dag treewidth} $\hsw(H)$ of $H$, which captures the relevant structure of $H$ when counting its copies in a $d$-degenerate graph.
In a simplified form, the bound is the following:
\begin{theorem}
\label{thm:hsw_ub}
For any $k$-node pattern $H$  one can compute \homo{H,G}, \sub{H,G}, and \ind{H,G} in time $f(k) \cdot O(d^{k-\hsw(H)} n^{\hsw(H)} \log n)$.
\end{theorem}
Let us briefly explain this result.
The heart of the problem is computing \homo{H,G}; once we know how to do this, we can obtain \sub{H,G} and \ind{H,G} via inclusion-exclusion arguments at the price of an extra multiplicative factor $f(k)$, like in~\cite{Borgs&2006,Curticapean&2017}.
To compute \homo{H,G}, we give $G$ an acyclic orientation with maximum outdegree $d$.
Then, we take every possible acyclic orientation $P$ of $H$, and compute \homo{P,G} where by \homo{P,G} we mean the number of homomorphisms of $P$ in $G$ that respect the orientations of the arcs.
Note that the number of such homomorphisms can be $n^{\Omega(k)}$ even if $G$ has bounded degeneracy (for example, if $P$ is an independent set), so we cannot list them explicitly.
At this point we introduce our technical tool, the \emph{dag tree decomposition} of $P$.
This is a tree $T$ that captures the relevant reachability relations between the nodes of $P$.
Given $T$, one can compute $\homo{P,G}$ via dynamic programming in time $f(k) \cdot O(d^{k-\hsw(T)} n^{\hsw(T)} \log n)$, where $\hsw(T) \in \{1,\ldots,k\}$ is the \emph{width} of $T$.
The dynamic program computes \homo{P,G} by combining carefully the homomorphism counts of certain subgraphs of $P$.
The dag-treewidth $\hsw(H)$, which is the parameter appearing in the bound of Theorem~\ref{thm:hsw_ub}, is the maximum width of the optimal dag tree decomposition of any acyclic orientation $P$ of any graph obtainable by identifying nodes of and/or adding edges to $H$ (this arises from the inclusion-exclusion arguments).
With this, our technical machinery is complete.
To obtain the bounds of the previous paragraph, we show how to compute efficiently dag tree decompositions of low width, and apply a more technical version of Theorem~\ref{thm:hsw_ub}.

We conclude by complementing Theorem~\ref{thm:hsw_ub} with a lower bound based on the Exponential Time Hypothesis.
This lower bound shows that in the worst case the dag-treewidth $\hsw(H)$ cannot be beaten, and therefore our decomposition captures, at least in part, the complexity of counting subgraphs in $d$-degenerate graphs.
\begin{theorem}
\label{thm:hsw_lb}
Under the Exponential Time Hypothesis~\cite{Impagliazzo&1998}, no algorithm can compute $\sub{H,G}$ or $\ind{H,G}$ in time $f(d,k) \cdot n^{o(\hsw(H)/\log{\hsw(H)})}$ for all $H$.
\end{theorem}

\subsection{Preliminaries and notation}
\label{sub:prelim}
Both $G=(V,E)$ and $H=(V_H,E_H)$ are simple graphs, possibly disconnected.
For any subset $V' \subseteq V$ we denote by $G[V']$ the subgraph of $G$ induced by $V'$; the same notation applies to any graph.
A \emph{homomorphism} from $H$ to $G$ is a map $\phi : V_H \to V$ such that $\{u,u'\} \in E_H$ implies $\{\phi(u),\phi(u')\} \in E$.
We write $\phi : H \to G$ to highlight the edges that $\phi$ preserves.
When $H$ and $G$ are oriented, $\phi$ must preserve the direction of the arcs.
If $\phi$ is injective then we have an injective homomorphism.
We denote by $\homo{H,G}$ and $\inj{H,G}$ the number of homomorphisms and injective homomorphisms from $H$ to $G$.
To avoid confusion, we will use the symbol $\psi$ to denote maps that are not necessarily homomorphisms.
The symbol $\simeq$ denotes isomorphism.
A \emph{copy} of $H$ in $G$ is a subgraph $F \subseteq G$ such that $F \simeq H$.
If moreover $F \simeq G[V_F]$ then $F$ is an induced copy.
We denote by $\sub{H,G}$ and $\ind{H,G}$ the number of copies and induced copies of $H$ in $G$; we may omit $G$ if clear from the context. 
When we give an acyclic orientation to the edges of $H$, we denote the resulting dag by $P$.
All the notation described above applies to directed graphs in the natural way.

The \emph{degeneracy} of $G$ is the smallest integer $d$ such that there is an acyclic orientation of $G$ with maximum outdegree bounded by $d$.
Such an orientation can be found in time $O(|E|)$ by repeatedly removing from $G$ a minimum-degree node~\cite{Nesetril2012sparsity}.
From now on we assume that $G$ has this orientation.
Equivalently, $d$ is the smallest integer that bounds from above the minimum degree of every subgraph of $G$.

We assume the following operations take constant time: accessing the $i$-th arc of any node $u \in V$, and checking if $(u,v)$ is an arc of $G$ for any pair $(u,v)$.
Our upper bounds still hold if checking an arc takes time $O(\log n)$, which can be achieved via binary search if we first sort the adjacency lists of $G$.
The $\log n$ factor in our bounds appears since we assume logarithmic access time for our dictionaries, each of which holds $O(n^k)$ entries.
This factor can be removed by using dictionaries with worst-case $O(1)$ access time (e.g., hash maps), at the price of obtaining probabilistic/amortized bounds rather than deterministic ones.

Finally, we recall the tree decomposition and treewidth of a graph.
For any two nodes $X,Y$ in a tree $T$, we denote by $T(X,Y)$ the unique path between $X$ and $Y$ in $T$.
\begin{definition}[see~\cite{Diestel2017}, Ch.\ 12.3]
\label{def:treedecomp}
\label{def:treewidth}
Given a graph $G=(V,E)$, a tree decomposition of $G$ is a tree $D=(V_D,E_D)$ such that each node $X \in V_D$ is a subset $X \subseteq V$, and that\,\footnote{Formally, we should define a tree together with a mapping between its nodes and the subsets of $V$. However, the definition adopted here is sufficient for our purposes and lightens the notation.}:
\begin{enumerate}
\item[1.] $\cup_{X \in V_D} X = V$ 
\item[2.] for every edge $e = \{u,v\} \in G$ there exists $X \in D$ such that $u,v \in X$
\item[3.] $\forall\, X, X', X'' \in V_T$, if $X \in D(X', X'')$ then $X' \cap X'' \subseteq X$
\end{enumerate}
The width of a tree decomposition $T$ is $\tw(T) = \max_{X \in V_T} |X| - 1$.
The treewidth $\tw(G)$ of a graph $G$ is the minimum of $\tw(T)$ over all tree decompositions $T$ of $G$.
\end{definition}

\subsection{Related work}
As anticipated, the fastest algorithm known for computing $\ind{H,G}$ is the one by Ne\v{s}et\v{r}il and Poljak~\cite{Nesetril&1985} that runs in time $O(n^{\omega\lfloor\frac{k}{3}\rfloor + (k \bmod 3)})$ where $\omega$ is the matrix multiplication exponent.
With the current bound $\omega \le 2.373$, this running time is in $O(n^{0.791 k + 2})$.
Unfortunately, the algorithm is based on fast matrix multiplication, which makes it oblivious to the sparsity of $G$.

Under certain assumptions on $G$, faster algorithms are known.
If $G$ has bounded maximum degree, $\Delta=O(1)$, then we can compute $\ind{H,G}$ in time $c^k \cdot O(n)$ for some $c=c(\Delta)$ via multivariate graph polynomials~\cite{Patel&2018}.
If $G$ has treewidth $\tw(G) \le k$, and we are given a tree decomposition of $G$ of such width, then we can compute $\ind{H,G}$ in time $2^{O(k \log k)} O(n)$; see Lemma 18.4 of~\cite{Nesetril2012sparsity}.
When $G$ is planar, we obtain an $f(k) \, O(n)$ algorithm where $f$ is exponential in $k$~\cite{Eppstein1995}.
All these assumptions are stronger than bounded degeneracy, and the techniques cannot be extended easily.
A more general class that captures all these cases is that of nowhere-dense graphs~\cite{nowhere-dense}, for which there exist fixed-parameter-tractable subgraph counting algorithms~\cite{Grohe&FO}.
Nowhere dense graphs however do not include all bounded degeneracy graphs or all graphs with bounded average degree.

Even assuming $G$ has bounded degeneracy, algorithms faster than Ne\v{s}et\v{r}il-Poljak are known only when $H$ belongs to special classes.
The earliest result of this kind is the classic algorithm by Chiba and Nishizeki~\cite{Chiba&1985} to list all $k$-cliques in time $O(d^{k-1} n)$.
Eppstein showed that one can list all maximal cliques in time $O(d 3^{d/3} n)$~\cite{Eppstein&2010} and all non-induced complete bipartite subgraphs in time $O(d^3 2^{2d} n)$~\cite{Eppstein1994}.
These algorithms exploit the degeneracy ordering of $G$ in a way similar to ours.
In fact, our techniques can be seen as a generalization of~\cite{Chiba&1985} that takes into account the structure of $H$.
We note that a fundamental limitation of~\cite{Chiba&1985,Eppstein1994,Eppstein&2010} is that they lists all the copies of $H$, which for a generic $H$ might be $\Theta(n^k)$ even if $G$ has bounded degeneracy (for example if $H$ is the independent set).
In contrast, we list the copies of subgraphs $H$, and combine them to infer the number of copies of $H$.
To be more precise, we list the homomorphisms of $H$, which is another difference we have with~\cite{Chiba&1985,Eppstein1994,Eppstein&2010} and a point we have in common with previous work~\cite{Curticapean&2017}.

Regarding our ``dag tree decomposition'', it is inspired to the standard notion of tree decomposition of a graph, and it yields a similar dynamic program.
Yet, the similarity between the two decompositions is rather superficial; indeed, our dag-treewidth can be $O(1)$ when the treewidth is $\Omega(k)$, and vice versa.
Our decomposition is unrelated to the several notions of tree decomposition for directed graphs already known~\cite{Ganian2010digraph}.
Finally, our lower bounds are novel; no general lower bounds in terms of $d$ and of the structure of $H$ was available before.

\subsection{Manuscript organisation.}
In Section~\ref{sec:simple} we build the intuition with a gentle introduction to our approach.
In Section~\ref{sec:dec} we give our dag tree decomposition and the dynamic program for counting homomorphisms.
In Section~\ref{sec:twbound} we show how to compute good dag tree decompositions.
Finally, in Section~\ref{sec:lb} we prove the lower bounds.

\section{Exploiting degeneracy orientations}
\label{sec:simple}
We build the intuition behind our approach, starting from the classic algorithm for counting cliques by Chiba and Nishizeki~\cite{Chiba&1985}.
The algorithm begins by orienting $G$ acyclically so that $\max_{v \in G}d_{\text{out}}(v) \le d$, which takes time $O(|E|)$.
With $G$ oriented acyclically, we take each $v \in G$ in turn, enumerate every subset of $(k-1)$ out-neighbors of $v$, and check its edges.
In this way we can explicitly find all $k$-cliques of $G$ in time $O(k^2 d^{k-1} n)$. 
Observe that the crucial fact here is that an acyclically oriented clique has exactly one \emph{source}, that is, a node with no incoming arcs.
We would like to extend this approach to an arbitrary pattern $H$.
Since every copy of $H$ in $G$ appears with exactly one acyclic orientation, we take every possible acyclic orientation $P$ of $H$, count the copies of $P$ in $G$, and sum all the counts.
Thus, the problem reduces to counting the copies of an arbitrary dag $P$ in our acyclic orientation of $G$.

Let us start in the naive way.
Suppose $P$ has $s$ sources.
Fix a directed spanning forest $F$ of $P$.
This is a collection of $s$ directed disjoint trees rooted at the sources of $P$ (arcs pointing away from the roots).
Clearly, each copy of $P$ in $G$ contains a copy of $F$.
Hence, we can enumerate the copies of $F$ in $G$, and for each one check if it is a copy of $P$.
To this end, first we enumerate the $O(n^s)$ possible $s$-uples of $V$ to which the sources of $P$ can be mapped.
For each such $s$-uple, we enumerate the possible mappings of the remaining $k-s$ nodes of the forest.
This can be done in time $O(d^{k-s})$ by a straightforward extension of the out-neighbor listing algorithm above.
Finally, for each mapping we check if its nodes induce $P$ in $G$, in time $O(k^2)$.
The total running time is $O(k^2 d^{k-s} n^s )$.
Unfortunately, if $P$ is an independent set then $s=k$ and the running time is $O(k^2n^k)$, so we have made no progress over the naive algorithm.

At this point we introduce our first idea.
For reference we use the toy pattern $P$ in Figure~\ref{fig:cycle}.
Instead of enumerating the copies of $P$ in $G$, we decompose $P$ into two \emph{pieces}, $P(1)$ and $P(3,5)$.
Here, $P(1)$ denotes the subgraph of $P$ reachable from $1$ (that is, the transitive closure of $1$ in $P$). The same for $P(3)$ and $P(5)$, and we let $P(3,5)=P(3) \cup P(5)$.
Now we count the copies of $P(1)$, and then the copies of $P(3,5)$, hoping to combine the result in some way to obtain the count of $P$.
To simplify the task, we focus on counting homomorphisms rather than copies (see below).
Thus, we want to compute $\hom(P,G)$ by combining $\hom(P(1),G)$ and $\hom(P(3,5),G)$.

Now, clearly, knowing $\hom(P(1),G)$ and $\hom(P(3,5),G)$ is not sufficient to infer $\hom(P,G)$.
Thus, we need to solve a slightly more complex problem.
For every pair $x,y \in V(G)$, let $\phi : \{2,6\} \mapsto V(G)$ be the map given by $\phi(2)=x$ and $\phi(6)=y$. We let $\homo{P, G, (x,y)}$ be the number of homomorphisms of $P$ in $G$ whose restriction to $\{2,6\}$ is $\phi$.
By a counting argument one can immediately see that:
\begin{align}
\label{eqn:sumhom}
\homo{P,G} = \sum_{\phi  : \{2,6\} \to V(G)} \!\!\!\!\!\! \homo{P, G,\phi}
\end{align}
Thus, to compute $\hom(P,G)$ we only need to compute $\homo{P, G,\phi}$ for all possible $\phi$.
Now, define $\homo{P(1),G, \phi}$ and $\homo{P(3,5), G,\phi}$ with the same meaning as above.
A crucial observation is that $\{2,6\}$, the domain of $\phi$, is precisely the set of nodes in $P(1) \cap P(3,5)$.
It is not difficult to see that this implies:
\begin{align}
\homo{P, G,\phi} = \homo{P(1),G,\phi} \cdot \homo{P(3,5),G, \phi}
\label{eqn:toy_homo}
\end{align}
Thus, now our goal is to compute $\homo{P(1),G,\phi}$ and $\homo{P(3,5),G, \phi}$ for all $\phi : \{2,6\} \to V(G)$.
To this end, we list all $\phi_{P(1)} : P(1) \to G$ with the technique above, and for each such $\phi_{P(1)}$ we increment a counter associated to $(\phi_{P(1)}(2),\phi_{P(1)}(6))$ in a dictionary with default value $0$.
Thus, we obtain $\homo{P(1),G,\phi}$ for all $\phi:\{2,6\} \to V$.
Since $P(1)$ has one source, we enumerate $O(k^2 n)$ maps.
If the dictionary takes time $O(\log n)$ to access an entry, the total running time is $O(k^2 n \log n)$.
The same technique applied to $P(3,5)$ yields a running time of $O(k^2 n^2 \log n)$, since $P(3,5)$ has two sources.
Finally, we apply Equation~\ref{eqn:sumhom} by running over all entries in the first dictionary and retrieving the corresponding value from the second dictionary.
The total running time is $O(k^2 n^2 \log n)$, while enumerating the homomorphisms of $P$ would have required time $O(k^2 n^3)$.
\begin{figure}[h]
\centering
\begin{tikzpicture}[
scale=1.08
]
\pgfdeclarelayer{foreground}
\pgfsetlayers{background,main,foreground}
\def \n {6}
\def \radius {1cm}
\def \margin {10} 
\begin{pgfonlayer}{foreground}
\foreach \s in {1,...,\n}
{
  \node[graph] (\s) at ({180 + 360/\n * (\s - 1)}:\radius) {\s};
}
\draw[->, >=latex] (1) -- (2);
\draw[->, >=latex] (1) -- (6);
\draw[->, >=latex] (3) -- (2);
\draw[->, >=latex] (3) -- (4);
\draw[->, >=latex] (5) -- (4);
\draw[->, >=latex] (5) -- (6);
\end{pgfonlayer}
\begin{pgfonlayer}{main}
\filldraw[stain] plot [smooth cycle, tension=1.1] coordinates {($(2)+(.2,-.3)$) ($(1)+(-.4,0)$) ($(6)+(.2,+.3)$)};
\filldraw[stain] plot [smooth cycle, tension=.8] coordinates {($(2)+(-.3,-.2)$) ($(3)+(0,-.4)$) ($(4)+(.3,0)$) ($(5)+(0,+.4)$) ($(6)+(-.3,+.2)$) ($(1)+(.5,0)$)};
\end{pgfonlayer}
\end{tikzpicture}
\caption{Toy example: an acyclic orientation $P$ of $H=C_6$, decomposed into two pieces.}
\label{fig:cycle}
\end{figure}

Let us abstract the general approach from this toy example.
We want to decompose $P$ into a set of pieces $P_1,P_2,\ldots$ with the following properties: (i) each piece $P_i$ has a small number of sources $s(P_i)$, and (ii) we can obtain $\homo{P,G,\phi}$ by combining the homomorphism counts of the $P_i$.
This is achieved by the \emph{dag tree decomposition}, which we introduce in Section~\ref{sec:dec}.
Like the tree decomposition for undirected graphs, the dag tree decomposition leads to a dynamic program to compute $\homo{P,G}$.

\section{DAG tree decompositions}
\label{sec:dec}
Let $P=(V_P,A_P)$ be a directed acyclic graph.
We denote by $\roots_P$, or simply $\roots$, the set of nodes of $P$ having no incoming arc.
These are the \textit{sources} of $P$.
We denote by $V_P(u)$ the transitive closure of $u$ in $P$, i.e.\ the set of nodes of $P$ reachable from $u$, and we let $P(u) = P[V_P(u)]$ be the corresponding subgraph of $P$.
For a subset of sources $B \subseteq \roots$ we let $V_P(B) = \cup_{u \in B} V_P(u)$ and $P(B) = P[V_P(B)]$.
Thus, $P(B)$ is the subgraph of $P$ induced by all nodes reachable from $B$.
We call $B$ a \textit{bag} of sources.
We can now formally introduce our decomposition.

\begin{definition}[Dag tree decomposition]
\label{def:piecedecomp}
Let $P=(V_P,A_P)$ be a dag.
A dag tree decomposition (d.t.d.) of $P$ is a (rooted) tree $T=(\vt,\et)$ with the following properties:
\begin{enumerate}\itemsep2pt
\item each node $B \in \vt$ is a bag of sources $B \subseteq \roots_P$ 
\item $\bigcup_{B \in \vt} B = \roots_P$
\item \label{pr:joint_path}for all $B,B_1,B_2 \in T$, if $B \in T(B_1,B_2)$ then $V_P(B_1) \cap V_P(B_2) \subseteq V_P(B)$ 
\end{enumerate}
\end{definition}
One can see the similarity with the tree decomposition of an undirected graph (Definition~\ref{def:treedecomp}).
However, our dag tree decomposition differs crucially in two aspects.
First, the bags are subsets of $\roots$ rather than subsets of $V_P$.
This is because the time needed to list the homomorphisms between $P(B_i)$ and $G$ is driven by $n^{|B_i|}$.
Second, the path-intersection property (3) concerns the pieces reachable from the bags rather than the bags themselves.
The reason is that, to combine the counts of two pieces together, their intersection must form a separator in $P$ (similarly to the tree decomposition of an undirected graph).
The dag tree decomposition induces the following notions of \emph{width}, used throughout the rest of the article.
\begin{definition}
\label{def:pw}
The \emph{width} of $T$ is $\sw(T) = \max_{B \in \vt} |B|$.
The \emph{dag treewidth} $\sw(P)$ of $P$ is the minimum of $\sw(T)$ over all dag tree decompositions $T$ of $P$.
\end{definition}
Clearly $\sw(P) \in \{1,\ldots, k\}$ for any $k$-node dag $P$.
Figure~\ref{fig:dtd} shows a pattern $P$ together with a d.t.d.\ of width $1$.
We observe that $\sw(P)$ has no obvious relation to the treewidth $\tw(H)$ of $H$; see the discussion in Subsection~\ref{sub:incexc}.

\begin{figure}[h]
\centering
\hfill
\begin{tikzpicture}[
scale=1.2
]
\pgfdeclarelayer{foreground}
\pgfsetlayers{background,main,foreground}

\def \n {3}
\def \radius {1}
\def \margin {10} 
\begin{pgfonlayer}{foreground}
\node[graph] (v1) at (0,0) {$1$};
\node[graph] (v1a) at ($(v1)+(30:1)$) {};
\node[graph] (v1b) at ($(v1)+(90:1)$) {};
\draw[->,>=latex] (v1) -- (v1a);
\draw[->,>=latex] (v1) -- (v1b);
\draw[->,>=latex] (v1b) -- (v1a);
\node[graph] (v2) at ($(v1a)+(60:1)$) {2};
\node[graph] (v2a) at ($(v1a)+(0:1)$) {};
\node[graph] (v3) at ($(v2a)+(60:1)$) {3};
\node[graph] (v3a) at ($(v2a)+(0:1)$) {};
\draw[->,>=latex] (v1a) -- (v2a);
\draw[->,>=latex] (v2a) -- (v3a);
\draw[->,>=latex] (v2) -- (v1a);
\draw[->,>=latex] (v2) -- (v2a);
\draw[->,>=latex] (v3) -- (v2a);
\draw[->,>=latex] (v3) -- (v3a);
\node[graph] (v4) at ($(v2a)+(-60:1)$) {4};
\node[graph] (v4a) at ($(v2a)+(-120:1)$) {};
\draw[->,>=latex] (v4) -- (v2a);
\draw[->,>=latex] (v4) -- (v4a);
\node[graph] (v5) at ($(v4a)+(-120:1)$) {5};
\node[graph] (v5a) at ($(v4a)+(-60:1)$) {};
\draw[->,>=latex] (v5) -- (v5a);
\draw[->,>=latex] (v5) -- (v4a);
\draw[->,>=latex] (v5a) -- (v4a);
\end{pgfonlayer}
\def \d {.35}
\def \t {.8}
\begin{pgfonlayer}{main}
\filldraw[stain] plot [smooth cycle, tension=\t] coordinates {($(v1)+(-120:\d)$) ($(v1b)+(120:\d)$) ($(v1a)+(0,\d)$) ($(v2a)+(60:\d)$) ($(v2a)+(-60:\d)$)  ($(v1a)+(0,-\d)$)};
\filldraw[stain] plot [smooth cycle, tension=\t] coordinates {($(v2)+(0,\d)$) ($(v1a)+(-150:\d)$) ($(v2a)+(-30:\d)$)};
\filldraw[stain] plot [smooth cycle, tension=\t] coordinates {($(v3)+(0,\d)$) ($(v2a)+(-150:\d)$) ($(v3a)+(-30:\d)$)};
\filldraw[stain] plot [smooth cycle, tension=\t] coordinates {($(v2a)+(0,\d)$) ($(v4a)+(-150:\d)$) ($(v4)+(-30:\d)$)};
\filldraw[stain] plot [smooth cycle, tension=\t] coordinates {($(v4a)+(0,\d)$) ($(v5)+(-150:\d)$) ($(v5a)+(-30:\d)$)};
\end{pgfonlayer}
\end{tikzpicture}
\hfill
\begin{tikzpicture}[
scale=1.2,
every node/.style={circle, inner sep=0pt, minimum size=18, fill=white},
line width=.7pt
]
\node[graph] (v1) at (0,0) {\{1\}};
\node[graph] (v2) at ($(v1)+(-1,-1)$) {\{2\}};
\node[graph] (v3) at ($(v1)+(0,-1)$) {\{3\}};
\node[graph] (v4) at ($(v1)+(1,-1)$) {\{4\}};
\node[graph] (v5) at ($(v4)+(0,-1)$) {\{5\}};
\draw (v1) edge (v2);
\draw (v1) edge (v3);
\draw (v1) edge (v4);
\draw (v4) edge (v5);
\node[opacity=0] (phant) at ($(v5)+(0,-.4)$) {};
\end{tikzpicture}
\hfill\mbox{}
\caption{Left: a dag $P$ formed by five pieces. Right: a dag tree decomposition $T$ for $P$. Since $\sw(T)=1$ and the largest piece contains $4$ nodes, we can compute $\homo{P,G}$ in time $O(d^{3} n \log n)$.}
\label{fig:dtd}
\end{figure}

\subsection{Counting homomorphisms via dag tree decompositions}
For any $B \in \vt$ let $T(B)$ be the subtree of $T$ rooted at $B$.
We let $\oline{B}$ be the down-closure of $B$ in $T$, that is, the union of all bags in $T(B)$.
Consider $P(\oline{B})$, the subgraph of $P$ induced by the nodes reachable from $\oline{B}$ (note the difference with $P(B)$, which contains only the nodes reachable from $B$).
We compute $\homo{P(\oline{B}),G}$ in a bottom-up fashion over all $B$, starting with the leaves of $T$ and moving towards the root.
This is similar to the dynamic program given by the standard tree decomposition (see~\cite{Fomin&2010}).

As anticipated, we actually compute $\homo{P(\oline{B}),\phi}$, the number of homomorphisms that extend a fixed mapping $\phi$.
We need the following concept:
\begin{definition}
Let $P_1=(V_{P_1},A_{P_1}), P_2=(V_{P_2},A_{P_2})$ be two subgraphs of $P$, and let $\phi_1: P_1 \to G$ and $\phi_2 : P_2 \to G$ be two homomorphisms.
We say $\phi_1$ and $\phi_2$ \emph{respect} each other if $\phi_1(u)=\phi_2(u)$ for all $u \in V_{P_1} \cap V_{P_2}$.
\end{definition}
Given some $\phi$, we denote by $\homo{P_1,G,\phi_2}$ the number of homomorphisms from $P_1$ to $G$ that respect $\phi_2$.
We can now present our main algorithmic result.
\begin{theorem}
\label{thm:counting}
Let $P$ be any $k$-node dag, and $T=(\vt,\et)$ be a d.t.d.\ for $P$.
Fix any $B \in \vt$ as the root of $T$.
There is a dynamic programming algorithm \treeCount($P, T, B$) that in time $O(|\vt| k^2  d^{k-\sw(T)} n^{\sw(T)} \log n)$ computes $\homo{P(\oline{B}), G, \phi_B}$ for all $\phi_B : P(B) \to G$.
This is also a bound on the time needed to compute $\homo{P,G}$.
\end{theorem}
The proof of Theorem~\ref{thm:counting} is given in the next subsection.
Before continuing, let $f_{T}(k)$ be an upper bound on the time needed to compute a d.t.d.\ of minimum width with at most $2^k$ bags for a pattern on $k$ nodes.
We can show that such a d.t.d.\ always exists:
\begin{lemma}
\label{lem:small_dtd}
Any $k$-node dag $P$ has a minimum-width d.t.d.\ on at most $2^{k}$ bags.
\end{lemma}
\begin{proof}
We show that, if a d.t.d.\ $T=(\vt,\et)$ has two bags containing exactly the same sources, then one of the two bags can be removed.
This implies that there exists a minimum-width d.t.d.\ where every bag contains a distinct source set, which therefore has at most $2^k$ bags.
Suppose indeed $T$ contains two bags $X$ and $X'$ formed by the same subset of sources.
Let $B$ be the neighbor of $X$ on the unique path $T(X,X')$.
Let $T^*=(\vt^*,\et^*)$ be the tree obtained from $T$ by replacing the edge $\{B',X\}$ with $\{B',B\}$ for every neighbor $B' \ne B$ of $X$ and then deleting $X$.
Clearly $|\vt^*|=|\vt|-1$, and properties (1) and (2) of Definition~\ref{def:piecedecomp} are satisfied.
Let us then check property (3).
Consider a generic path $T^*(B_1,B_2)$ and look at the corresponding path $T(B_1,B_2)$.
If $T(B_1,B_2)$ does not contain edges that we deleted, then $T^*(B_1,B_2)=T(B_1,B_2)$.
In this case the property holds for any bag in $T^*(B_1,B_2)$ since it holds in $T$.
Suppose instead $T(B_1,B_2)$ contains edges that we deleted.
Then $T^*(B_1,B_2)$ contains the same bags of $T(B_1,B_2)$ save that $X$ is replaced by $B$.
Thus we only need to check that $V_P(B_1) \cap V_P(B_2) \subseteq V_P(B)$.
By property (3), $V_P(B_1) \cap V_P(B_2) \subseteq V_P(X)$.
Moreover, since by construction $B \in T(X,X')$, property (3) also gives $V_P(X) = V_P(X) \cap V_P(X') \subseteq V_P(B)$.
Thus $V_P(B_1) \cap V_P(B_2) \subseteq V_P(B)$.
Therefore $T^*$ is a d.t.d.\ for $P$.
\qed
\end{proof}
Then, as an immediate corollary of Theorem~\ref{thm:counting}, we have:
\begin{theorem}
\label{thm:homo_cost}
We can compute $\homo{P,G}$ in time $f_{T}(k) + O(k^2 2^k d^{k-\sw(P)} n^{\sw(P)} \log n)$.
\end{theorem}
Theorem~\ref{thm:homo_cost} will be used in Section~\ref{sub:incexc} to prove the bounds for our original problem of counting the copies of $H$ via inclusion-exclusion arguments.

\subsubsection{Proof of Theorem~\ref{thm:counting}}
The algorithm behind Theorem~\ref{thm:counting} is similar to the one for counting homomorphisms using a tree decomposition.
To start, we prove that our dag tree decomposition enjoys a separator property similar to the one enjoyed by tree decompositions.
\begin{lemma}
\label{lem:separator}
Let $T$ be a rooted d.t.d.\ and let $B_1, \ldots, B_l$ be the children of $B$ in $T$.
Then for all $i \in [l]$:
\begin{enumerate}\itemsep0pt
\item[a.] $V_P(\oline{B_i}) \cap V_P(\oline{B_j}) \subseteq V_P(B)$ for all $j \ne i$
\item[b.] for any arc $(u,u') \in P(\oline B)$, if $u \in V_P(\oline{B_i}) \setminus V_P(B)$ then $u' \in V_P(\oline{B_i})$
\item[c.] for any arc $(u',u) \in P(\oline B)$, if $u \in V_P(\oline{B_i}) \setminus V_P(B)$ then $u' \in V_P(\oline{B_i})$
\end{enumerate}
\end{lemma}
\begin{proof}
We prove (a).
Suppose for some $ i \ne j$ we have $V_P(\oline{B_i}) \cap V_P(\oline{B_j}) \nsubseteq V_P(B)$.
So there exists some node $u \in V_P$ such that $u \in V_P(\oline{B_i})$, $u \in V_P(\oline{B_j})$, and $u \notin V_P(B)$.
By definition of $\oline{\cdot}$, this implies $u \in V_P(B_i')$ and $u \in V_P(B_j')$ for some bags $B_i' \in T(B_i)$ and $B_j' \in T(B_j)$.
Observe however that $B \in T(B_i',B_j')$.
Thus, by point (3) of Definition~\ref{def:piecedecomp}, we have $u \in V_P(B)$.
This contradicts the third inclusion, $u \notin V_P(B)$.

Now  we prove (b) and (c).
For (b), since $u \in V_P(\oline{B_i})$ and $(u,u') \in P$, then $u' \in V_P(\oline{B_i})$ too.
For (c), suppose by contradiction $u' \notin V_P(\oline{B_i})$.
Therefore, either $u' \in V_P(B)$, or $u' \in V_P(B_j')$ for some $B_j' \in \oline{B_j}$ with $j \ne i$.
In both cases however we have $u \in V_P(B)$: in the first case this holds since $u$ is reachable from $u'$, and in the second case since $B \in T(B_i,B_j')$ and by point (3) of Definition~\ref{def:piecedecomp}.
Thus in any case $u \in V_P(B)$, which contradicts again $u \notin V_P(B)$.
\qed
\end{proof}
Lemma~\ref{lem:separator} says that $V_P(B)$ is a separator for the sub-patterns $P(\oline{B_i})$ in $P$.
This allows us to compute $\homo{P(\oline{B})}$ by combining $\homo{P(\oline{B_1})},\ldots,\homo{P(\oline{B_l})}$.

Next, we show that each homomorphism $\phi$ of $P(\oline{B})$ is the \emph{juxtaposition} (definition below) of some $\phi_B$ of $B$ and some $\phi_1,\ldots,\phi_l$ of $\oline{B_1},\ldots,\oline{B_l}$, provided they respect $\phi_B$.
This establishes a bijection, implying that we can count the homomorphisms $\phi$ by multiplying the counts of the homomorphisms $\phi_B,\phi_1,\ldots,\phi_l$.

\begin{definition}
\label{def:juxt}
Let $\{\phi_1,\ldots,\phi_{\ell}\}$ be any set of homomorphisms, where for all $i=1,\ldots,\ell$ we have $\phi_i : X_i \to G$ and $\phi_i$ respects $\phi_j$ for all $j =1,\ldots,\ell$.
The \emph{juxtaposition} of $\phi_1,\ldots,\phi_{\ell}$, denoted by $\phi_1\ldots\phi_\ell$, is the homomorphism $\phi : \cup_{i=1}^{\ell} X_i \to G$ such that $\phi(u)=\phi_i(u)$ whenever $u \in X_i$.
\end{definition}
Note that the juxtaposition is always well-defined ad unique, since the $\phi_i$ respect each other and the image of every $u$ is determined by at least one among $\phi_1,\ldots,\phi_{\ell}$.
\begin{lemma}
\label{lem:homo_comp}
Let $T$ be a d.t.d.\ and let $B_1, \ldots, B_l$ be the children of $B$ in $T$.
Fix any $\phi_B : P(B) \to G$.
Let $\Phi(\phi_B) = \{\phi : P(\oline{B}) \to G \,|\, \phi \text{ respects } \phi_B\}$, and for $i=1,\ldots,l$ let $\Phi_i(\phi_B) = \{\phi : P(\oline{B_i}) \to G \,|\, \phi \text{ respects } \phi_B\}$.
Then there is a bijection between $\Phi(\phi_B)$ and $\Phi_1(\phi_B) \times \ldots \times \Phi_l(\phi_B)$, and therefore:
\begin{align}
\homo{P(\oline{B}),G,\phi_B} = \prod_{i=1}^{l}\homo{P(\oline{B_i}),G,\phi_B}
\end{align}
\end{lemma}
\begin{proof}
First, we show there is an injection between $\Phi(\phi_B)$ and $\Phi_1(\phi_B) \times \ldots \times \Phi_l(\phi_B)$.
Fix any $\phi \in \Phi(\phi_B)$, and consider the tuple $(\phi_1, \ldots, \phi_l)$ where
each $\phi_i$ is the restriction of $\phi$ to $P(\oline{B_i})$.
Note that $\phi_i$ is unique, and that it respects $\phi_B$ since $\phi$ does.
Thus $\phi_i \in \Phi_i(\phi_B)$.
It follows that $(\phi_1, \ldots, \phi_l) \in \Phi_1(\phi_B) \times \ldots \times \Phi_l(\phi_B)$.
Now we show there is an injection between $\Phi_1(\phi_B) \times \ldots \times \Phi_l(\phi_B)$ and $\Phi(\phi_B)$.
Consider any tuple $(\phi_1, \ldots, \phi_l) \in \Phi_1(\phi_B) \times \ldots \times \Phi_l(\phi_B)$, and consider the juxtaposition $\phi = \phi_B\phi_1\ldots\phi_l$.
Then $\phi : P(\oline{B}) \to G$ and $\phi$ respects $\phi_B$.
It follows that $\phi \in \Phi(\phi_B)$.
\qed
\end{proof}
Last, we bound the cost of enumerating the homomorphisms of a piece of $P$.
\begin{lemma}
\label{lem:listing2}
Given any $B \subseteq \roots$, the set of homomorphisms $\Phi = \{ \phi : P(B) \to G \}$ has size $O(d^{k-|B|} n^{|B|})$ and can be enumerated in time $O(k^2 d^{k-|B|} n^{|B|})$.
\end{lemma}
\begin{proof}
We prove the bound on the enumeration time; the proof gives immediately also the bound on $|\Phi|$.
Let $B=\{u_1,\ldots,u_b\}$ where $b = |B|$.
Fix a spanning forest $\{T_1, \ldots, T_b\}$ of $P(B)$, where each $T_i=(V_i,A_i)$ is a directed tree rooted at $u_i$ (arcs pointing away from the root).
Consider any $\phi \in \Phi$, and let $\phi_i$ be its restriction to $V_i$.
Clearly, $\phi=\phi_1\ldots\phi_{b}$.
Note that $\phi_i$ is a homomorphism of $T_i$ in $G$.
Thus, to enumerate $\Phi$ we can enumerate each possible tuples $(\phi_1,\ldots,\phi_b)$ where $\phi_i$ is a homomorphism of $T_i$ in $G$ for all $i$.
Note that not all such tuples give a valid juxtaposition that is a homomorphism $\phi \in \Phi$. However, we can check if $\phi \in \Phi$ in time $O(k^2)$ by checking the arcs between the images of $\phi$ in $G$.

Let then $\Phi_{T_i}$ be the set of homomorphisms of $T_i$ in $G$.
We show how to enumerate $\Phi_{T_i}$ in time $O(d^{|V_i|-1} n)$, and thus all tuples $(\phi_1,\ldots,\phi_b) \in \Phi_{T_1} \times \ldots \times \Phi_{T_b}$ in time $\prod_{i=1}^b O(d^{|V_i|-1} n) = O(d^{k-b} n^{b})$.
Together with the check on the arcs, this gives a total running time of $O(k^2 d^{k-b} n^{b})$ for enumerating $\Phi$, as desired.
To enumerate $\Phi_{T_i}$, we take each $v \in G$ and enumerate all $\phi_i \in \Phi_{T_i}$ such that $\phi_i(s_i)=v$.
To this end note that, once we have fixed $\phi_i(x)$, for each arc $(x,y) \in T_i$ we have at most $d$ choices for $\phi(y)$.
Thus we can enumerate all $\phi_i \in \Phi_{T_i}$ that map $s_i$ to $v$ in time $d^{|V_i|-1}$.
The total time to enumerate $\Phi_{T_i}$ is therefore $O(d^{|V_i|-1} n)$, as claimed.
\qed
\end{proof}

We can now describe our dynamic programming algorithm, \treeCount, to compute $\homo{P(\oline{B}),G}$.
Given a d.t.d.\ $T$ of $P$, the algorithm goes bottom-up from the leaves of $T$ towards the root, combining the counts using Lemma~\ref{lem:homo_comp}.
For readability, we write the algorithm in a recursive fashion.
\begin{algorithm}[h!]
\caption{\treeCount($P, T, B$)}
\begin{algorithmic}[1]
\State $C_{B} =$ empty dictionary with default value $0$
\If{$B$ is a leaf} \Comment{base case}
\For{every homomorphism $\phi_B : P(B) \to G$} \label{ln:enum1}
\State $C_B(\phi_B) = 1$ \label{ln:acc1}
\EndFor
\Else
\State let $B_1,\ldots,B_l$ be the children of $B$ in $T$
\For{$i=1,\ldots,l$} 
\State $C_{B_i} =$ \treeCount($P, T, B_i$) \label{ln:rec} \Comment{recur on the subtree of $T$ rooted at $B_i$}
\State $AGG_{B_i} =$ empty dictionary with default value $0$
\For{every key $\phi$ in $C_{B_i}$} \label{ln:reindex} \Comment{we aggregate the values of $C$}
  \State let $\phi_r$ be the restriction of $\phi$ to $V_P(B) \cap V_P(\oline{B_i})$
  \State $AGG_{B_i}(\phi_r) \pluseq C_{B_i}(\phi)$
\EndFor \label{ln:reindex2}
\EndFor
\For{every homomorphism $\phi : P(B) \to G$} \label{ln:enum2} \Comment{Lemma~\ref{lem:homo_comp} in action}
\State for $i=1,\ldots,l$ let $\phi_r$ be the restriction of $\phi$ to $V_P(B) \cap V_P(\oline{B_i})$
\State $C_B(\phi) = \prod_{i=1}^l AGG_{B_i}(\phi_{r})$ \label{ln:acc2}
\EndFor
\EndIf
\State \Return $C_B$
\end{algorithmic}
\end{algorithm}
We prove:
\begin{lemma}
\label{lem:pccorrect}
\label{lem:pccost}
Let $P$ be any dag, $T=(\vt,\et)$ any d.t.d.\ for $P$, and $B$ any element of $\vt$.
\treeCount($P, T, B$) in time $O(|\vt| \poly(k) d^{k-\sw(T)} n^{\sw(T)} \log{n})$ returns a dictionary $C_B$ that for all $\phi_B : P(B) \to G$ satisfies $C_B(\phi_B) = \homo{P(\oline{B}), G, \phi_B}$.
\end{lemma}
\begin{proof}
We first prove the correctness, by induction on the nodes of $T$.
The base case is when $B$ is a leaf of $T$.
In this case $P(B) = P(\oline{B})$, and the algorithm sets $C_B(\phi_B) = 1$ for each $\phi_B : P(B) \to G$.
Therefore $C_B(\phi_B) = \homo{P(\oline{B}), G, \phi_B}$ as desired.
The inductive case is when $B$ is an internal node of $T$.
As inductive hypothesis we assume that, for every child $B_i$ of $B$, the dictionary $C_{B_i}$ computed at line~\ref{ln:rec} satisfies $C_{B_i}(\phi) = \homo{P(\oline{B_i}),G,\phi}$ for every $\phi : P(B_i) \to G$.
Let $\Phi_{P(\oline{B_i})}$ be the set of homomorphisms from $P(\oline{B_i})$ to $G$, and let $\Phi_{P(\oline{B_i})}(\phi)$ be the subset of elements of $\Phi_{P(\oline{B_i})}$ that respect $\phi$.
Thus, the inductive hypothesis says that $C_{B_i}(\phi) = |\Phi_{P(\oline{B_i})}(\phi)|$ for every $\phi : P(B_i) \to G$.

Now consider the loop at lines~\ref{ln:reindex}--\ref{ln:reindex2}.
We claim that, after that loop, we have:
\begin{align}
\label{eqn:AGG}
AGG_{B_i}(\phi_r) 
= \!\!\sum_{\substack{\phi \text{ in } C_{B_i} \\ \phi \text{ respects } \phi_{r}}}
\!\!\!\!\!\!\! \!\!\!\! C_{B_i}(\phi) 
= \sum_{\substack{\phi \,:\, P(B_i) \to G \\ \phi \text{ respects } \phi_{r}}}
\!\!\!\!\!\!\! \!\!\!\! C_{B_i}(\phi) 
= \sum_{\substack{\phi \,:\, P(B_i) \to G \\ \phi \text{ respects } \phi_{r}}}
\!\!\!\!\!\!\! \!\!\!\! \big|\Phi_{P(\oline{B_i})}(\phi)\big|
= \big|\Phi_{P(\oline{B_i})}(\phi_r)\big|
\end{align}
The first equality holds since the loop adds $C_{B_i}(\phi)$  to $AGG_{B_i}(\phi_r)$ if and only if the restriction of $\phi$ to $V_P(B) \cap V_P(\oline{B_i})$ is $\phi_r$, that is, if and only if $\phi$ respects $\phi_r$.
The second equality holds since the keys of $C_{B_i}$ are a subset of $\{\phi \,:\, P(B_i) \to G\}$ and $C_{B_i}(\phi)=0$ if $\phi$ is not in $C_{B_i}$.
The third equality holds by the inductive hypothesis above.
The fourth equality holds since the sets $\Phi_{P(\oline{B_i})}(\phi)$ form a partition of $\Phi_{P(\oline{B_i})}(\phi_r)$.

Finally, consider the loop at lines~\ref{ln:enum2}--\ref{ln:acc2}.
We claim that line~\ref{ln:acc2} sets:
\begin{align}
C_B(\phi)
= \prod_{i=1}^l \big|\Phi_{P(\oline{B_i})}(\phi_r)\big|
= \prod_{i=1}^l \big|\Phi_{P(\oline{B_i})}(\phi)\big| 
= \homo{P(\oline{B}), G, \phi}
\end{align}
The first equality holds by coupling line~\ref{ln:acc2} and Equation~\eqref{eqn:AGG}.
The second equality holds since any element of $\Phi_{P(\oline{B_i})}$ respects $\phi$ if and only if it respects its restriction $\phi_r$ to $V_P(B) \cap V_P(\oline{B_i})$, thus $\Phi_{P(\oline{B_i})}(\phi_r) = \Phi_{P(\oline{B_i})}(\phi)$.
The last equality holds by definition of $\Phi_{P(\oline{B_i})}(\phi_r)$ and by Lemma~\ref{lem:homo_comp}.
This proves the correctness.

Let us turn to the running time.
We can represent a homomorphism $\phi$ as a tuple of nodes of $G$.
Now, if $B$ is a leaf in $T$, then the running time is dominated by the loop at line~\ref{ln:enum1}.
By the bound of Lemma~\ref{lem:listing2}, the loop performs $O(d^{k-|B|}n^{|B|})$ iterations.
Each iteration takes time $O(k^2)$ to check $\phi_B$ (see again Lemma~\ref{lem:listing2}), and time $O(k \log{n})$ to update $C_B(B, \phi)$ since $C_B$ contains $O(n^k)$ entries.
This gives a running time of $O(k^2 d^{k-|B|}  n^{|B|} \log{n})$.
If $B$ is an internal node of $T$, then the time taken by Lines~\ref{ln:reindex}--\ref{ln:reindex2} is dominated by the recursive calls at line~\ref{ln:rec}.
The loop at line~\ref{ln:enum2} follows the analysis above for each of the $l$ children of $B$, for a total running time of $O(k^2 l\, d^{k-|B|} n^{|B|} \log{n})$.
The total running time excluding recursive calls is then $O(k^2 l\, d^{k-|B|} n^{|B|} \log{n})$ as well.
The thesis follows by recursing on the subtrees of $T(B)$, noting that the sum of the number of children of all bags is at most $|\vt|$. 
\qed
\end{proof}

\subsection{Inclusion-exclusion arguments and the dag-treewidth of undirected graphs}
\label{sub:incexc}
We turn to computing $\homo{H,G}$, $\sub{H,G}$ and $\ind{H,G}$.
We do so via standard inclusion-exclusion arguments, using our algorithm for computing $\homo{P,G}$ as a primitive.
To this end we shall define appropriate notions of width for undirected pattern graphs.
Let $\Sigma(H)$ be the set of all dags $P$ that can be obtained by orienting $H$ acyclically.
Let $\Theta(H)$ be the set of all equivalence relations on $V_H$ (that is, all the partitions of $V_H$), and for $\theta \in \Theta(H)$ let $H/\theta$ be the pattern obtained from $H$ by identifying equivalent nodes according to $\theta$ and removing loops and multiple edges.
Let $D(H)$ be the set of all supergraphs of $H$ on the node set $V_H$, including $H$.
\begin{definition}
\label{def:widths}
The \emph{dag treewidth} of $H$ is $\hsw(H) = \sw_3(H)$, where:
\begin{align}
\sw_1(H) &= \max\{ \sw(P) : {P \in \Sigma(H)} \}\\
\sw_2(H) &= \max\{ \sw_1(H/\theta) : {\theta \in \Theta(H)} \}\\
\sw_3(H) &= \max\{ \sw_2(H') : {H' \in D(H)} \}
\end{align}
\end{definition}
Note that $\sw(H)$ is unrelated to the treewidth $\tw(H)$.
For example, when $H$ is a clique we have $\tw(H)=k$ and $\sw(H)=1$; when $H$ is the independent set we have $\tw(H)=1$ and $\sw(H)=\Theta(k)$, see Lemma~\ref{thm:pw_tw}; and when $H$ is an expander we have $\tw(H),\sw(H) \in \Theta(k)$, see again Lemma~\ref{thm:pw_tw}.
In fact, $\sw(H)$ is within constant factors of the independence number $\alpha(H)$ of $H$ (see Section~\ref{sub:alpha}), and thus decreases as $H$ becomes denser.
This happens because adding arcs increases the number of nodes reachable from the sources of $P \in \Sigma(H)$, so we may need fewer sources to reach a given piece of $P$.
When $H$ is a clique, $P$ is reachable from just one source and thus $\sw(H)=1$.

Clearly, $\sw_1(H) \le \sw_2(H) \le \sw(H)$.
The intuition behind $\sw_1(H)$ is that, in $G$, each homomorphism of $H$ corresponds to a homomorphism of some acyclic orientation $P$ of $H$.
Thus to compute $\homo{H,G}$ we sum $\homo{P,G}$ over all orientation $P$ of $H$, and the running time is dominated by the $P$ with largest treewidth.
The intuition behind $\sw_2(H)$ is similar but now we look at computing $\sub{H,G}$.
Since homomorphisms can map different nodes of $H$ to the same node of $G$, to recover $\sub{H,G}$ we must combine $\homo{H',G}$ for all possible $H' = H/\theta$ through inclusion-exclusion arguments.
The intuition behind $\sw_3(H)$ is that to compute $\ind{H,G}$ we must remove from $\sub{H,G}$ the counts of $\sub{H',G}$ for certain supergraphs $H'$ of $H$.
Indeed, the three measures $\sw_1(H),\sw_2(H),\sw(H)$ yield:
\begin{theorem}
\label{thm:wrapping}
Consider any $k$-node pattern graph $H=(V_H,E_H)$, and let $f_{T}(k)$ be an upper bound on the time needed to compute a d.t.d.\ of minimum width on $2^{O(k \log k)}$ bags for any $k$-node dag.
Then one can compute:
\begin{itemize}
\item $\homo{H,G}$ in time $2^{O(k \log k)} \cdot O(f_{T}(k) + d^{k-\sw_1(H)} n^{\sw_1(H)} \log n)$,
\item $\sub{H,G}$ in time $2^{O(k \log k)} \cdot O(f_{T}(k) + d^{k-\sw_2(H)} n^{\sw_2(H)} \log n)$,
\item $\ind{H,G}$ in time $2^{O(k^2)} \cdot O(f_{T}(k) + d^{k-\sw(H)} n^{\sw(H)} \log n)$.
\end{itemize}
The claim still holds if we replace $\tau_1,\tau_2,\sw$ with upper bounds, and $f_{T}(k)$ with the time needed to compute a d.t.d.\ on $2^{O(k \log k)}$ bags that satisfies those upper bounds.
\end{theorem}

\begin{proof}
We prove the three bounds in three separated steps. The last claim follows straightforwardly.

\paragraph{From dags to undirected patterns.}
Let $H$ be any undirected pattern.
First, note that:
\begin{align}
\label{eq:homo_sum}
\homo{H,G} = \sum_{P \in \Sigma(H)} \!\!\!\!\homo{P,G}
\end{align}
Let indeed $\Phi(H)=\{\phi : H \to G\}$ be the set of homomorphisms from $H$ to $G$.
Similarly, for any $P\in \Sigma(H)$ define $\Phi(P)=\{\phi_P : P \to G\}$ (note that $\phi_P$ must preserve the direction of the arcs).
Then, there is a bijection between $\Phi(H)$ and $\cup_{P \in \Sigma(H)} \Phi(P)$.
Consider indeed any $\phi \in \Phi(H)$.
Let $\sigma$ be the orientation of $H$ that assigns to $\{u,v\} \in E_H$ the orientation of $\{\phi(u),\phi(v)\}$ in $G$, and let $P=H_{\sigma}$.
Then $\phi$ is a homomorphism of $P$ in $G$.
On the other hand consider any homomorphism $\phi \in \Phi(P)$ for some acyclic orientation $P$ of $H$.
By ignoring the orientation of the edges, $\phi \in \Phi(H)$, too.

Thus, to compute $\homo{H,G}$ we compute $\homo{P,G}$ for all  $P \in \Sigma(H)$ and apply Equation~\ref{eq:homo_sum}.
Clearly, enumerating $|\Sigma(H)|$ takes time $O(k!) = 2^{O(k \log k)}$.
For each $P$, by Lemma~\ref{lem:small_dtd} in time $f_{T}(k)$ we compute a d.t.d.\ $T=(\vt,\et)$ of width $\tau(P)$ such that $|\vt|\le 2^k$.
Then, by Lemma~\ref{lem:pccost} we compute $\homo{P,G}$ in time $O(2^k \poly(k) d^{k-\tau(P)} n^{\tau(P)} \log n)$.
Thus, we can compute every $\homo{P,G}$ in time $O(f_{T}(k) + 2^{O(k)} d^{k-\tau(P)} n^{\tau(P)} \log n)$.
Multiplying by $2^{O(k \log k)}$ gives the first bound of the theorem.

\paragraph{From homomorphisms to non-induced copies.}
Recall that $H/\theta$ is the graph obtained from $H$ by identifying the nodes in the same equivalence class and removing loops and multiple edges, where $\theta \in \Theta(H)$ is an equivalence relation (or partition) over $V_H$.
Then, by Equation~15 of~\cite{Borgs&2006}:
\begin{align}
\label{eqn:inj}
\inj{H,G} = \sum_{\theta \in \Theta(H)} \!\!\!\! \mu(\theta) \, \homo{H/\theta, G}
\end{align}
where $\mu(\theta) = \prod_{A \in \theta} (-1)^{|A|-1}(|A|-1)!$, where $A$ runs over the equivalence classes (the sets) in $\theta$.
Thus, to compute $\inj{H,G}$, we enumerate all $\theta \in \Theta(H)$, compute $\homo{H/\theta, G}$, and apply Equation~\ref{eqn:inj}.
It is known that $|\Theta| = 2^{O(k \ln{k})}$ (see e.g.~\cite{Berend&2010}), and clearly for each $\theta$ we can compute $\mu(\theta)$ and $H/\theta$ in $O(\poly(k))$.
Thus, the first bound of the theorem holds for computing $\inj{H,G}$ too.
Finally, we compute $\sub{H,G} = \frac{\inj{H,G}}{\aut{H}}$, where $\aut{H}$ is the number of automorphisms of $H$, which can be computed in time $2^{O(\sqrt{k \ln k})}$~\cite{Mathon1979}.
This proves the second bound of the theorem.

\paragraph{From non-induced to induced.}
Finally, let $D(H)$ be the set of all supergraphs of $H$ on the same node set.
Then from Equation~14 of \cite{Borgs&2006}:
\begin{align}
\label{eqn:ind}
\ind{H,G} = \sum_{H' \in D(H)} \!\!\!\!(-1)^{|E_{H'} \setminus E_H|} \, \inj{H',G}
\end{align}
To compute $\ind{H,G}$, we take every $H'\in D(H)$, compute $\inj{H',G}$, and apply Equation~\ref{eqn:ind}.
Since $|D(H)| \le 2^{k^2}$, the third bound of the theorem follows.
\qed
\end{proof}

The algorithmic part of our work is complete.
We shall now focus on computing good dag tree decompositions, so to instantiate Theorem~\ref{thm:wrapping} and obtain the upper bounds of Section~\ref{sub:results}.

\section{Computing good dag tree decompositions}
\label{sec:twbound}
In this section we show how to compute dag tree decompositions of low width.
First, we show that for every $k$-node dag $P$ we can compute in time $2^{O(k)}$ a dag tree decomposition $T$ that satisfies $\hsw(T) \le \lfloor \frac{k}{4}\rfloor+2$.
This result requires a nontrivial proof. As a corollary, we prove Theorem~\ref{thm:mainbound} and Theorem~\ref{thm:lowavgd}.
Second, we give improved bounds for cliques minus $\epsilon$ edges; as a corollary, we prove Theorem~\ref{thm:dense_ub}.
Third, we give improved bounds for complete multipartite graphs plus $\epsilon$ edges; as a corollary, we prove Theorem~\ref{thm:bipartite}.
Finally, we show that $\Omega(\alpha(H)) \le \hsw(H) \le \alpha(H)$ where $\alpha(H)$ is the independence number of $H$, which is of independent interest.
This implies that the trivial decomposition on one bag has width that is asymptotically optimal.

To proceed, we need some additional notation.
For a dag $P$, we say $v \in V_P$ is a \emph{joint} if it is reachable from at least two sources, i.e., if $v \in V_P(u) \cap V_P(u')$ for some $u,u' \in \roots$ with $u \ne u'$.
Let $\mkey$ be the set of joints of $P$.
We write $\mkey(u)$ for the set of joints reachable from $u$, and for any $X \subseteq V_P$ we let $\mkey(X) = \cup_{u \in X} \mkey(u)$.
Similarly, we denote by $\roots(y)$ the sources from which $y$ is reachable, and we let $\roots(X) = \cup_{u \in X} \roots(u)$.

\subsection{A bound for all patterns}
\label{sub:general_twbound}
This subsection is devoted to prove: 
\begin{restatable}{theorem}{widthbound}
\label{thm:widthbound}
For any dag $P=(V_P,A_P)$, in time $O(1.7549^k)$ we can compute a dag tree decomposition $T=(\vt,\et)$ with  $\sw(T) \le \min(\lfloor\frac{e}{4}\rfloor,\lfloor\frac{k}{4}\rfloor) + 2$ and $|\vt|=O(k)$, where $k=|V_P|$ and $e=|A_P|$.
\end{restatable}
\noindent By combining Theorem~\ref{thm:widthbound} with Definition~\ref{def:widths} and Theorem~\ref{thm:wrapping},  we obtain as a corollary Theorem~\ref{thm:mainbound}.
The proof of Theorem~\ref{thm:widthbound} is divided in four steps, as follows.
First (step 1), we remove the ``easy'' pieces of $P$; this can break $P$ into several connected components, and we show that their d.t.d.'s can be composed into a d.t.d.\ for $P$.
Next, we show that if the $i$-th component has $k_i$ nodes, then it admits a d.t.d.\ of width $\frac{k_i}{4}+2$.
This requires to ``peel'' the component to remove its tree-like parts (step 2) and decomposing the remainder using a reduction to standard tree decompositions (step 3).
Finally, we wrap up our results and conclude the proof (step 4).

Throughout the proof, the relevant structure of $P$ is encoded by a graph that we call the \emph{skeleton} of $P$, defined as follows.
\begin{definition}
The \emph{skeleton} of a dag $P=(V_P,A_P)$ is the bipartite dag $\skel{P} = (V_{\sk},E_{\sk})$ where $V_{\sk} = \roots \cup \mkey$ and $E_{\sk} \subseteq \roots \times \mkey$, and $(u,v) \in E_{\sk}$ if and only if $v \in J(u)$.
\end{definition}
\noindent 
Figure~\ref{fig:tree_decomp} gives an example.
Note that $\skel{P}$ does not contain nodes that are neither sources nor joints, as they are irrelevant to the d.t.d..
Note also that computing $\skel{P}$ takes time $O(\poly(k))$.

\begin{figure}[h]
\centering
\resizebox{.9\textwidth}{!}{%
\begin{tikzpicture}[
scale=0.6
]
\node[graph] (r1) at (0,0) {1};
\node[graph] (r2) at (2,0) {2};
\node[graph] (r3) at (4,0) {3};
\node[graph] (r4) at (6,0) {4};
\node[graph] (v1) at ($(r1)+(-1,-2)$) {5};
\node[graph] (v2) at ($(v1)+(2,0)$) {6};
\node[graph] (v3) at ($(v2)+(2,0)$) {7};
\node[graph] (v4) at ($(v3)+(2,0)$) {8};
\node[graph] (v5) at ($(v4)+(2,0)$) {9};
\path[->] (r1) edge (v1);
\path[->] (r1) edge (v2);
\path[->] (r2) edge (v2);
\path[->] (r2) edge (v3);
\path[->] (r2) edge (v4);
\path[->] (r3) edge (v3);
\path[->] (r3) edge (v4);
\path[->] (r3) edge (v5);
\path[->] (r4) edge (v5);
\path[->] (v4) edge (v3);
\end{tikzpicture}
\hspace*{40pt}
\begin{tikzpicture}[
scale=0.6
]
\node[source] (r1) at (0,0) {1};
\node[source] (r2) at (2,0) {2};
\node[source] (r3) at (4,0) {3};
\node[source] (r4) at (6,0) {4};
\node[joint] (j6) at ($(r1)+(0,-2)$) {6};
\node[joint] (j7) at ($(j6)+(2,0)$) {7};
\node[joint] (j8) at ($(j7)+(2,0)$) {8};
\node[joint] (j9) at ($(j8)+(2,0)$) {9};
\path[->] (r1) edge (j6);
\path[->] (r2) edge (j6);
\path[->] (r2) edge (j7);
\path[->] (r2) edge (j8);
\path[->] (r3) edge (j7);
\path[->] (r3) edge (j8);
\path[->] (r3) edge (j9);
\path[->] (r4) edge (j9);
\end{tikzpicture}
}
\caption{Left: a dag $P$. Right: its skeleton $\skel{P}$ (sources $\roots$ above, joints $\mkey$ below).}
\label{fig:tree_decomp}
\end{figure}

Let us now delve into the proof.
For any node $x$, we denote by $d_x$ the current degree of $x$ in the skeleton.

\medskip
\textbf{Step 1: Greedy bag construction.}
Set $B^{(0)}=\emptyset$ and let $\sk^{(0)} = (V_{\sk}^{(0)}, E_{\sk}^{(0)}) = \sk$.
Set $j=0$ and proceed iteratively as follows, recalling that $V_{\sk}^{(j)} = \roots_{\sk}^{(j)}\cup \mkey_{\sk}^{(j)}$.
For any source $u \in \roots_{\sk}^{(j)}$ let $V_{\sk}^{(j)}(u)$ be the transitive closure of $u$ in $\sk^{(j)}$.
If there exists a source $u \in \roots_{\sk}^{(j)}$ such that $|V_{\sk}^{(j)}(u)| \ge 4$, then let $B^{(j+1)}=B^{(j)} \cup \{u\}$, and let $\sk^{(j+1)}$ be obtained from $\sk^{(j)}$ by removing $V_{\sk}^{(j)}(u)$ from $V_{\sk}^{j}$.
Repeat the procedure until $|V_{\sk}^{(j)}(u)| \le 3$ for all $u$.
Suppose the procedure stops at $j=j^*$, producing the subset $B^*=B^{(j^*)}$ and the residual skeleton $\sk^* = \sk^{(j^*)} = (V_{\sk}^*, E_{\sk}^*)$. 
\begin{lemma}
\label{lem:bstar}
$|B^*| \le \min\big( \frac{k-|V_{\sk}^*|}{4}, \frac{e- |E_{\sk}^*|}{4}\big)$, where $k=|V_P|$ and $e=|A_P|$.
\end{lemma}
\begin{proof}
For the first term of the $\min()$, note that at each step we remove at least $4$ nodes from $\sk^{(j)}$ and add one node to $B^{(j)}$.
Hence $4|B^*| \le |V_{\sk} \setminus V_{\sk}^*| \le k-|V_{\sk}^*|$, which implies $|B^{*}| \le \frac{k-|V_{\sk}^{*}|}{4}$.
For the second term of the $\min()$, we show that the set of nodes $V_{\sk}^{(j)}(u)$ removed at step $j$ identifies at least $4$ unique arcs of $P$.
To this end, consider the sub-pattern $P^{(j)} = P \setminus P(B^{(j)})$ containing all nodes not reachable from $B^{(j)}$.
Note that $\sk^{(j)}$ is the skeleton of $P^{(j)}$.
Indeed, if $v \in {\mkey}^{(j)}$ then $v$ is not reachable from any source in $B^{(j)}$, since otherwise $v$ would have been removed before step $j$.
Thus, at step $j$ we are removing at least $3$ joint nodes of $P^{(j)}(u)$.
Therefore, $P^{(j)}(u)$ contains at least three arcs pointing to its joints.
In addition, by definition, the joints of $P^{(j)}(u)$ must be reachable from some node in $P^{(j)} \setminus P^{(j)}(u)$.
Thus there is an arc from $P^{(j)} \setminus P^{(j)}(u)$ to a node of $P^{(j)}(u)$, and this node is therefore a joint itself.
Thus, $P^{(j)}$ contains at least $4$ arcs pointing to the joints of $P^{(j)}(u)$.
Since the joints of $P^{(j)}(u)$ are then removed from $P^{(j)}$, these arcs are counted only once.
Hence $e \ge 4|B^{*}| + |E_{\sk}^{*}|$, and $|B^{*}| \le \frac{e- |E_{\sk}^{*}|}{4}$.
\qed
\end{proof}

Now, if $B^{*}=\roots$, then $T=(\{B^{*}\}, \emptyset)$ is a d.t.d.\ of $P$ whose width is $\sw(T) = |B^{*}|$.
By Lemma~\ref{lem:bstar}, $|B^{*}| \le \min(\frac{k}{4}, \frac{e}{4})$, so Theorem~\ref{thm:widthbound} is proven.

If instead $B^{*} \subset \roots$, then $\sk^*$ has $\ell \ge 1$ nonempty connected components.
For each $i=1,\ldots,\ell$ let $\sk_i=(\roots_i \cup \mkey_i, E_i)$ be the $i$-th component of $\sk^*$.
Let $P^* = P \setminus P(B^*)$, and let $P_i = P^{*}(\roots_i)$.
Then, $\sk_i$ is the skeleton of $P_i$; this follows from the same argument used in the proof of Lemma~\ref{lem:bstar}.
We shall now see that we can obtain a d.t.d.\ for $P$ by arranging the d.t.d.'s of the $P_i$ into a tree, and adding $B^*$ to all bags.

\begin{lemma}
\label{lem:compose}
For each $i=1,\ldots,\ell$ let $T_i=(\vt_i,\et_i)$ be a d.t.d.\ of $P_i$.
Consider the tree $T$ obtained as follows.
The root of $T$ is the bag $B^{*}$, and the subtrees below $B^{*}$ are $T_1,\ldots,T_{\ell}$, where each bag $B \in T_i$ has been replaced by $B \cup B^{*}$.
Then $T=(\vt,\et)$ is a d.t.d.\ of $P$ with $\sw(T) \le |B^{*}| + \max_{i=1,\ldots,\ell}\sw(T_i)$ and $|\vt|=1+\sum_{i=1}^{\ell}|\vt_i|$.
\end{lemma}
\begin{proof}
The claims on $\sw(T)$ and $|\vt|$ are straightforward.
Let us check that $T$ is a d.t.d.\ of $P$, via Definition~\ref{def:piecedecomp}.
Property (1) is immediate.
For property (2), note that $\cup_{B \in \vt_i} = \roots_i$ because $T_i$ is by hypothesis a d.t.d.\ of $\sk_i$. Thus $\cup_{B \in \vt} = B^{*} \cup (\cup_{i=1}^\ell \roots_i) = \roots_P$.
We turn to property (3).
Choose any two bags $B' \cup B^{*}$ and $B'' \cup B^{*}$ of $T$, where $B' \in T_i$ and $B'' \in T_j$ for some $i,j\in \{1,\ldots,\ell\}$, and any bag $B \cup B^{*} \in T(B' \cup B^{*}, B'' \cup B^{*})$.
Suppose first $i=j$; thus by construction $B \in T(B',B'')$.
Since $T_i$ is a d.t.d., then $\mkey_i(B') \cap \mkey_i(B'') \subseteq \mkey_i(B)$, and in $T$ this implies $V_P(B' \cup B^{*}) \cap V_P(B'' \cup B^{*}) \subseteq V_P(B \cup B^{*})$.
Suppose instead $i \ne j$.
Thus $\mkey_i(\roots_i) \cap \mkey_j(\roots_j) = \emptyset$ and this means that $\mkey(\roots_i) \cap \mkey(\roots_j) \subseteq \mkey(B^{*})$.
But $V_P(B_i) \cap V_P(B_j) \subseteq \mkey(\roots_i) \cap \mkey(\roots_j)$ and $\mkey(B^{*}) \subseteq V_P(B^{*})$, thus $V_P(B_i) \cap V_P(B_j) \subseteq V_P(B^{*})$.
It follows that for every bag $B \cup B^{*}$ of $T$ we have $V_P(B_i \cup B^{*}) \cap V_P(B_j \cup B^{*}) \subseteq V_P(B \cup B^{*})$.
\qed
\end{proof}

\medskip
\textbf{Step 2: Peeling $\sk_i$.}
We now remove the tree-like parts of $\sk_i$.
These include, for instance, sources that have only one reachable joint.
For each such source, we create a dedicated bag which becomes the child of another bag that reaches the same joint.
This removes a source without increasing the width of the decomposition.

The construction is recursive.
Let $P_i^{(0)} = P_i$ and $\sk_i^{(0)}=(\roots_i^{(0)} \cup \mkey_i^{(0)}, E_{i}^{(0)})=\sk_i$.
Set $j=0$.
For any node $x \in \sk_i^{(j)}$, we denote by $d_u^{(j)}$ its degree in $\sk_i^{(j)}$.
We will show that the tree $T^{(0)}$ returned by our recursive construction is a d.t.d.\ for $P_i^{(0)}=P_i$.

The base case is $|\roots_i^{(j)}|=1$.
In this case we set $T_i^{(j)}=(\{\roots_i^{(j)}\}, \emptyset)$.
Clearly, $T_i^{(j)}$ is a d.t.d.\ for $P_i^{(j)}$ of width $1$ and we are done.
Suppose instead $|\roots_i^{(j)}|>1$.
Recall that $d_u^{(j)} \le 2$ for all $u \in \roots_i^{(j)}$.
Consider the first one of these three cases that applies (if none of them does, then we stop):
\begin{enumerate}\itemsep3pt
\item $\exists\, u \in S_i^{(j)} \,:\, d_u^{(j)} = 1$. Then choose any such $u$, and we choose any $u' \in S_i^{(j)} \setminus \{u\}$ with $\mkey_i^{(j)}(u) \cap \mkey_i^{(j)}(u') \ne \emptyset$.
\item $\mkey_i^{(j)}(u)=\mkey_i^{(j)}(u')$ for some $u, u' \in \roots_i^{(j)}$ with $u \ne u'$. Then, choose any such $u,u'$. 
\item $\exists\, v \in \mkey_i^{(j)} \,:\, d_v^{(j)} = 1$. Then choose any such $v$, let $u$ be the unique source such that $v \in \mkey_i^{(j)}(u)$, and let $u' \ne u$ be any source with $\mkey_i^{(j)}(u) \cap \mkey_i^{(j)}(u') \ne \emptyset$.
\end{enumerate}
Then, we define $T_i^{(j)}$ recursively as follows.
Let $P_i^{(j+1)} = P_i(\roots_i^{(j)} \setminus \{u\})$, and let $\sk_i^{(j+1)}$ be the skeleton graph obtained from $\sk_i^{(j)}$ by removing $u$ and (for the third case) the node $v \in \mkey_i^{(j)}(u)$ that is reachable only from $u$.
We invoke the procedure recursively on $\sk_i^{(j+1)}$.
Suppose the recursive procedure returns a d.t.d.\ $T_i^{(j+1)}$ of $P_i^{(j+1)}$.
Then, $T_i^{(j+1)}$ must contain a bag $B'$ such that $u' \in B'$.
Create the bag $B_{u}=\{u\}$, set it as a child of $B'$ in $T_i^{(j+1)}$, and let the resulting tree be $T_i^{(j)}$.
Let us check that $T_i^{(j)}$ is a d.t.d.\ for $P_i^{(j)}$.
Properties (1) and (2) of Definition~\ref{def:piecedecomp} are obviously satisfied.
For property (3), since $B_u$ is a leaf,  we only need to check that $\mkey_i^{(j)}(B_u) \cap \mkey_i^{(j)}(B'') \subseteq \mkey_i^{(j)}(B')$ for all $B'' \in T_i^{(j)}$ and all $B' \in T_i^{(j)}(B_u,B'')$.
To this end note that, by the choice of $u$ and $u'$, for any $u'' \in \roots_i^{(j)} \setminus \{u,u'\}$ we have $\mkey_i^{(j)}(u) \cap \mkey_i^{(j)}(u'') \subseteq \mkey_i^{(j)}(u')$.
We repeat the entire procedure until we reach the base case, or until $|\roots_i^{(j)}|>1$ and none of the three cases above holds, in which case we move to the next phase.

Before continuing, we make sure that the procedure above is well defined; we must guarantee that, in each of the three cases, the node $u'$ exists.
One can see that $u'$ exists whenever $|\roots_i^{(j)}| > 1$ (which is true by hypothesis) and $\sk_i^{(j)}$ is connected.
To see that $\sk_i^{(j)}$ is connected, note that if this was not the case then in some step $h < j$ we removed a source $u$ with $d_u^{(h)} = 2$ such that no other source $u'$ has $\mkey_i^{(h)}(u)=\mkey_i^{(h)}(u')$.
However, this cannot happen by construction of the procedure.

\begin{figure}[h]
\centering
\resizebox{.87\textwidth}{!}{%
\begin{tikzpicture}[
scale=0.45
]
\def\x{2}
\def\xx{3}
\def\y{3}
\node[source] (r1) at (0,0) {$u_1$};
\node[source] (r0) at ($(r1)-(\x,0)$) {$u_0$};
\node[source] (r2) at ($(r1)+(\x,0)$) {$u_2$};
\node[source] (r3) at ($(r2)+(\x,0)$) {$u_3$};
\node[source] (r4) at ($(r3)+(\x,0)$) {$u_4$};
\node[source] (r5) at ($(r4)+(\x,0)$) {$u_5$};
\node[source] (r6) at ($(r5)+(\x,0)$) {$u_6$};
\node[source] (r7) at ($(r6)+(\x,0)$) {$u_7$};
\node[source] (r8) at ($(r7)+(\x,0)$) {$u_8$};
\node[source] (r9) at ($(r8)+(\x,0)$) {$u_9$};
\node[source] (r10) at ($(r9)+(\x,0)$) {$u_{10}$};
\node[joint] (v1) at (0,-\y) {1};
\node[joint] (v2) at ($(v1)+(\xx,0)$) {2};
\node[joint] (v3) at ($(v2)+(\xx,0)$) {3};
\node[joint] (v4) at ($(v3)+(\xx,0)$) {4};
\node[joint] (v5) at ($(v4)+(\xx,0)$) {5};
\node[joint] (v6) at ($(v5)+(\xx,0)$) {6};
\path[->] (r0) edge (v1);
\path[->] (r1) edge (v1);
\path[->] (r1) edge (v2);
\path[->] (r2) edge (v1);
\path[->] (r2) edge (v2);
\path[->] (r3) edge (v1);
\path[->] (r3) edge (v3);
\path[->] (r4) edge (v2);
\path[->] (r4) edge (v3);
\path[->] (r5) edge (v3);
\path[->] (r5) edge (v5);
\path[->] (r6) edge (v4);
\path[->] (r6) edge (v5);
\path[->] (r7) edge (v5);
\path[->] (r7) edge (v6);
\path[->] (r8) edge (v5);
\path[->] (r8) edge (v6);
\path[->] (r9) edge (v6);
\path[->] (r10) edge (v6);
\end{tikzpicture}
}
\\\vspace*{12pt}
\resizebox{.67\textwidth}{!}{%
\begin{tikzpicture}[
scale=0.45
]
\def\x{2}
\def\xx{3}
\def\y{3}
\node[source] (r1) at (0,0) {$u_1$};
\node[source] (r2) at ($(r1)+(\xx,0)$) {$u_3$};
\node[source] (r4) at ($(r2)+(\xx,0)$) {$u_4$};
\node[joint] (v1) at (0,-\y) {1};
\node[joint] (v2) at ($(v1)+(\xx,0)$) {2};
\node[joint] (v3) at ($(v2)+(\xx,0)$) {3};
\path[->] (r1) edge (v1);
\path[->] (r1) edge (v2);
\path[->] (r2) edge (v1);
\path[->] (r2) edge (v3);
\path[->] (r4) edge (v2);
\path[->] (r4) edge (v3);
\end{tikzpicture}
\hspace*{50pt}
\begin{tikzpicture}[
scale=0.4,
auto,
node distance=2.6cm
]
\def\x{2}
\def\xx{3}
\def\y{2}
\node[joint] (v1) at (0,0) {1};
\node[joint] (v2) at ($(v1)+(\xx,\y)$) {2};
\node[joint] (v3) at ($(v1)+(\xx,-\y)$) {3};
\path[-] (v1) edge node {$u_1$} (v2);
\path[-] (v1) edge[below left] node {$u_3$} (v3);
\path[-] (v2) edge node {$u_4$} (v3);
\end{tikzpicture}
}
\caption{Above: example of a skeleton component $\sk_i$. Below: the core $\skic$ obtained from $\sk_i$ after peeling (left), and its encoding as $C_i$ (right).}
\label{fig:bigpatt}
\end{figure}

\medskip

\textbf{Step 3: Decomposing the core.}
Suppose the peeling phase stops at $j=j^*$.
Let $\Pic=P_i^{(j^*)}$ and $\skic=(\rootsc_i \cup \mkeyc_i, \edgc_{i})=\sk_i^{(j^*)}$.
We say $\Pic$ is the \emph{core} of $P_i$; this is the part that determines the dag treewidth.
Now, since $\skic$ violates all three conditions of the peeling step, we have $d_u^{\bul}=2$ for every source $u$ and $d_v^{\bul} \ge 2$ for every joint $v$.
Thus $\skic$ can be encoded as a simple graph.
Formally, let $\core=(V_{\core}, E_{\core})$ where $V_{\core}=\mkeyc_i$ and $E_{\core} = \{e_u : u \in \rootsc_i\}$, where $e_u=\mkeyc_i(u)$ for each $u \in \rootsc_i$.
To ease the discussion, for the edges we use $u$ and $e_u$ interchangeably.
Figure~\ref{fig:bigpatt} gives an example.
Note that $\skic$ is the skeleton of $\Pic$, since $\rootsc_i$ are the sources of $\Pic$ and the degree bound above implies that each $v \in \mkey_c$ is reachable from at least two sources of $\Pic$.
In what follows we let $k_i = |\rootsc_i \cup \mkeyc_i| = |E_{\core}|+|V_{\core}|$.

We use $\core$ to compute a good d.t.d.\ $\Tc_i=(\vtc_i,\etc_i)$ of $\Pic$ via tree decompositions.
First, we show that four our purposes it is sufficient to find a d.t.d.\ of width at most $\frac{|E_{\core}|}{5} + 3$.
\begin{lemma}
\label{lem:core1}
If $\sw(\Pic) \le \frac{|E_{\core}|}{5} + 3$ then $\sw(\Pic) \le \frac{k_i}{4}+2$.
\end{lemma}
\begin{proof}
First, suppose that $|V_{\core}| \le 4$.
Since all nodes of $\core$ have degree at least $2$, then $\core$ contains a $4$-cycle, and thus an edge cover $B^{\text{cov}}$ of size $2$.
We then build $\Tc_i$ by setting $B^{\text{cov}}$ as root, and $B_u = \{u\}$ for every $u \in E_{\core} \setminus B^{\text{cov}}$ as child of $B^{\text{cov}}$.
This is clearly a d.t.d.\ for $\Pic$ of width $2 < \frac{k_i}{4}+2$, and thus $\sw(\Pic) < \frac{k_i}{4}+2$.

Suppose instead that $|V_{\core}| \ge 5$.
Note that $|E_{\core}| \ge |V_{\core}|$ by construction of $\core$.
One can check that these conditions imply $\frac{|E_{\core}|}{20}+\frac{|V_{\core}|}{4} > 1$, which in turn gives:
\begin{align}
\sw(\Pic) \le \frac{|E_{\core}|}{5} + 3 \le \frac{|E_{\core}|+|V_{\core}|}{4}+2 = \frac{k_i}{4}+2
\end{align}
concluding the proof.\qed
\end{proof}
Therefore, we compute a d.t.d.\ of width at most $\frac{|E_{\core}|}{5} + 3$.
We do this in two steps.
\begin{lemma}
\label{lem:core2}
In time $O(1.7549^{k_i})$ one can compute a tree decomposition $D=(V_D,E_D)$ for $\core$ with treewidth at most $\frac{|E_{\core}|}{5} + 2$ and $O(k_i)$ bags.
\end{lemma}
\begin{proof}
By Theorem 2 of~\cite{Kneis&2005}, the treewidth of a graph $G = (V, E)$ is at most $\frac{|E|}{5} + 2$.
By Theorems~5.23-5.24 of~\cite{Fomin&2010}, we can compute a minimum-width tree decomposition of an $n$-node graph in time $O(1.7549^{n})$.
By Lemma~5.16 of~\cite{Fomin&2010}, in time $O(n)$ we can transform such a decomposition into one that contains at most $4n$ bags, leaving its width unchanged.
Therefore, in time $O(1.7549^{k_i})$ we can build a tree decomposition $D$ for $\core$ with $O(k_i)$ bags that satisfies $\tw(D) \le \frac{|E_{\core}|}{5} + 2$.
\qed
\end{proof}
\begin{lemma}
\label{lem:core}
Let $D=(V_D,E_D)$ be a tree decomposition of $\core$.
In time $\poly(k)$ we can build a d.t.d.\ $\Tc_i=(\vtc_i,\etc_i)$ for $\Pic$ such that $\sw(\Tc_i) \le \tw(D)+1$ and $|\vtc_i|=O(|V_D|+k_i)$.
\end{lemma}
\begin{proof}
To simplify the notation let us write $T,\vt,\et$ in place of $\Tc_i,\vtc_i,\etc_i$.
We first show how to build the tree $T$.
The tedious part is proving it is a valid d.t.d.\ for $\Pic$.

The intuition is that $D$ covers the edges of $\core$, which correspond to the sources of $\Pic$.
This gives a way to ``convert'' the bags of $D$ into bags for $\Tc_i$.
For every $v \in V_{\core}$ choose an arbitrary incident edge $u_v=\{v,z\}\in E_{\core}$.
Replace each bag $Y \in D$ by $B(Y) \!=\! \{u_v \!:\! v \!\in\! Y\}$, and for every $u \in \rootsc_i \setminus \cup_{Y \in D} B(Y)$, choose a bag $B(Y) : J(u) \subseteq Y$, and set the bag $B_u = \{u\}$ as child of $B(Y)$.
Let $T$ be the resulting tree.
To see that the construction is well-defined, note that, by point (2) of Definition~\ref{def:treedecomp}, for any $u \in E_{\core}$ there exists some $Y \in D$ such that $u=\{x,y\} \subseteq Y$.
Therefore assigning $B_u$ as child of some $B(Y)$ with $u \subseteq Y$ is licit.
Now, $\sw(T) \le \tw(D)+1$ follows immediately by the facts that $|B(Y)| \le |Y|$ for all $Y \in D$ and that $|B_u|=1$ for each of the bags $B_u$ above, and by Definition~\ref{def:treewidth} and~\ref{def:piecedecomp}.
The bound $|\vtc_i|=O(|V_D|+k_i)$ holds since $T$ contains a bag for each bag of $D$, plus at most one bag for each node in $\rootsc_i$, and $|\rootsc_i| \le k_i$.

Let us then check that $T$ is a d.t.d.\ for $\Pic$ via Definition~\ref{def:piecedecomp}.
Clearly, $T$ is a tree and satisfies property (1).
For property (2), let $E_{\core}(D) = \cup_{Y \in D} B(Y)$.
Observe that by construction $\cup_{B \in T}B = E_{\core}(D) \cup (\cup_{u \in E_{\core} \setminus E_{\core}(D)} B_u)$.
The right-hand expression is $E_{\core}$.

It remains to check property (3).
First, if we have set $B_u$ as child of $B_Y$ then by construction $J(B_u) \subseteq J(B_Y)$.
Thus we can ignore any such $B_u$ and focus on the remaining bags of $T$, proving that every $B,B',B''$ such that $B \in T(B',B'')$ satisfy $J(B') \cap J(B'') \subseteq J(B)$.
Let $Y,Y',Y''$ the three bags of $D$ from which the construction produced respectively $B,B',B''$.
Observe that $B \in T(B',B'')$ implies $Y \in D(Y',Y'')$.
Now suppose that, by contradiction, there exists $v \in J(B')\cap J(B'')$ such that $v \notin J(B)$.
Note that, by construction, we must have put some $u'$ with $e_{u'} = \{v,z'\}$ in $B'$ and some $u''$ with $e_{u''} = \{v,z''\}$ in $B''$, for some $z',z'' \in V_i$.
Moreover, $Y' \cap \{v,z'\} \ne \emptyset$ and $Y'' \cap \{v,z''\} \ne \emptyset$, else we could not have $u' \in B'$ and $u'' \in B''$.
Finally, bear in mind that $v \notin Y$ and $u',u'' \notin B$, otherwise $v \in J(B)$, contradicting the hypothesis.
Now we consider three cases.
We use repeatedly properties (2) and (3) of Definition~\ref{def:treedecomp}.

\textbf{Case 1.} $v \in Y' \cap Y''$. Then $v \in Y$, a contradiction.

\textbf{Case 2.} $v \in Y'$ and $v \notin Y''$.
Then $z'' \in Y''$ and $u''$ with $e_{u''} = \{v,z''\}$ is the edge chosen to cover $z''$, else we would not put $u'' \in B_{Y''}$.
Moreover there must be $\hat{Y} \in D$ such that $e_{u''} = \{v,z''\} \subseteq \hat{Y}$.
For the sake of the proof root $D$ at $Y$, so $Y'$ and $Y''$ are in distinct subtrees.
If $\hat{Y}$ and $Y''$ are in the same subtree then $Y \in D(Y', \hat{Y})$, but $v \in Y' \cap \hat{Y}$ and thus $v \in Y$, a contradiction.
Otherwise $Y \in D(Y'', \hat{Y})$, and since $z'' \in Y'' \cap \hat{Y}$ then $z'' \in Y$ and then $u'' \in B(Y)$, a contradiction.

\textbf{Case 3.} $v \notin Y'$ and $v \notin Y''$.
Then $z' \in Y'$, $z'' \in Y''$, and $u',u''$ are the sources chosen to cover respectively $z',z''$.
Moreover there must be $\hat{Y}, \hat{Y}' \in D$ such that $e_{u'}=\{v,z'\} \subseteq \hat{Y}$ and $e_{u''}=\{v,z''\} \subseteq \hat{Y}'$.
Root again $D$ at $Y$.
If $Y \in D(\hat{Y}, \hat{Y}')$ then since $v \in \hat{Y} \cap \hat{Y}'$ it holds $v \in Y$, a contradiction.
Otherwise $\hat{Y}, \hat{Y}'$ are in the same subtree of $D$.
If the subtree is the same as $Y''$, then $Y \in D(Y', \hat{Y})$, but $z \in Y' \cap \hat{Y}$ and thus $z' \in Y$ and thus $u' \in B(Y)$, a contradiction.
Otherwise we have $Y \in D(Y'', \hat{Y})$; but $z'' \in Y'' \cap \hat{Y}$, thus $z'' \in Y$ and $u'' \in B(Y)$, again a contradiction.
\qed
\end{proof}
Combining Lemma~\ref{lem:core1}, Lemma~\ref{lem:core2}, and Lemma~\ref{lem:core}, we obtain:
\begin{lemma}
\label{lem:core3}
In time $O(1.7549^{k_i} \poly(k_i))$ we can compute a d.t.d.\ $\Tc_i=(\vtc_i,\etc_i)$ for $\Pic$ such that $\sw(\Tc_i) \le \lfloor\frac{k_i}{4}\rfloor+2$ and $|\vtc_i|=O(k_i)$. 
\end{lemma}
With Lemma~\ref{lem:core3}, we are almost done.
It remains to wrap all our bounds together.

\medskip
\textbf{Step 4: Assembling the tree.}
Recall the sub-patterns $P_i$ obtained after the greedy bag construction (step 1).
Let $T_i=(\vt_i,\et_i)$ be the d.t.d.\ for $P_i$ as returned by the recursive peeling and the core decomposition.
Since the peeling phase only adds bags of size $1$, then $\sw(T_i) = \sw(\Tc_i)$.
Therefore, by Lemma~\ref{lem:core2}, $\sw(T_i) \le \lfloor\frac{k_i}{4}\rfloor+2$.
Moreover, since each bag added in the peeling phase corresponds to a unique source, then $|\vt_i| = O(k_i+|V(P_i)|) = O(|V(P_i)|)$.

Let now $T=(\vt,\et)$ be the d.t.d.\ for $P$ obtained by assembling the trees $T_1,\ldots,T_{\ell}$ as in Lemma~\ref{lem:compose}.
By Lemma~\ref{lem:compose} itself, $\sw(T) \le |B^{*}| + \max_{i=1,\ldots,\ell}\sw(T_i)$, thus:
\begin{align}
\sw(T) \le |B^{*}| + \max_{i=1,\ldots,\ell} \Big\lfloor\frac{k_i}{4}\Big\rfloor + 2
\end{align}
Now, by Lemma~\ref{lem:bstar} we know that $P(B^*)$ has at least $4|B^{*}|$ nodes and $4|B^{*}|$ arcs.
Similarly, since each $\sk_i$ has at least $k_i$ nodes and $k_i$ arcs, then $P \setminus P(B^{*})$ has at least $\sum_{i=1}^{\ell}k_i$ nodes and $\sum_{i=1}^{\ell}k_i$ arcs.
Then $\sw(T) \le \lfloor\frac{k}{4}\rfloor + 2$ and $\sw(T) \le \lfloor\frac{e}{4}\rfloor + 2$, so $\sw(T) \le \min(\lfloor\frac{k}{4}\rfloor, \lfloor\frac{e}{4}\rfloor) + 2$.
Moreover, $|\vt|=1 + \sum_{i=1}^{\ell}|\vt_i| = O(\sum_{i=1}^{\ell} |V(P_i)|) = O(k)$.
Finally, by Lemma~\ref{lem:core2} the time to build $T_i$ is $O(1.7549^{k_i} \poly(k_i))$, since the peeling phase clearly takes time $\poly(k_i)$.
The total time to build $T$ is therefore $O(1.7549^{k} \poly(k))$.
This concludes the proof of Theorem~\ref{thm:widthbound}.

\subsection{Bounds for quasi-cliques (Theorem~\ref{thm:dense_ub})}
\label{sub:dense_bound}
\begin{lemma}
\label{lem:quasicliques}
If a $k$-node dag $P$ has ${k \choose 2} - \epsilon$ edges, then in time $O(\poly(k))$ one can compute a d.t.d.\ $T$ for $P$ on two bags such that $\sw(T) \le \lceil \nicefrac{1}{2} + \sqrt{\nicefrac{\epsilon}{2}} \, \rceil$.
\end{lemma}
\begin{proof}
The source set $\roots$ of $P$ is an independent set.
Hence $\epsilon \ge {|\roots| \choose 2}$, and $|\roots| \le 1 + \sqrt{2\epsilon}$.
Consider any tree $T$ on two bags $B_1,B_2$ such that $B_1 \cup B_2 = \roots$, $|B_1|=\lfloor |\roots|/2 \rfloor$, and $|B_2|=\lceil |\roots|/2 \rceil$.
It is immediate to check that $T$ satisfies the claim.
\qed
\end{proof}
By coupling Lemma~\ref{lem:quasicliques} and Theorem~\ref{thm:wrapping}, for computing $\homo{H,G}$ and $\sub{H,G}$ we obtain a running time bound of $2^{O(k \log k)} \cdot O(d^{k-\lceil \frac{1}{2} + \sqrt{\frac{\epsilon}{2}} \, \rceil} n^{\lceil \frac{1}{2} + \sqrt{\frac{\epsilon}{2}} \, \rceil} \log n)$.
For  $\ind{H,G}$, we refine the bound of Theorem~\ref{thm:wrapping} by observing that $|D(H)| \le 2^{\epsilon}$.
This yields a running time bound of $2^{O(\epsilon + k \log k)} \cdot O(d^{k-\lceil \frac{1}{2} + \sqrt{\frac{\epsilon}{2}} \, \rceil} n^{\lceil \frac{1}{2} + \sqrt{\frac{\epsilon}{2}} \, \rceil} \log n)$.
This concludes the proof of Theorem~\ref{thm:dense_ub}.

\subsection{Bounds for quasi-multipartite graphs (Theorem~\ref{thm:bipartite})}
\begin{lemma}
\label{lem:hbipartite}
If $H$ is a complete multipartite graph, then $\sw_2(H) = 1$.
If $H$ is a complete multipartite graph plus $\epsilon$ edges, then $\sw_2(H) \le \lfloor\frac{\epsilon}{4}\rfloor+2$.
In any case, for any $\theta \in \Theta(H)$, for any acyclic orientation $P$ of $H/\theta$ we can compute in time $2^{O(k)}$ a d.t.d.\ of $P$ on $O(k)$ bags whose width satisfies the bounds above.
\end{lemma}
\begin{proof}
First, suppose $H=(V_H,E_H)$ is complete multipartite, so $V_H = V_H^1 \cup \ldots \cup V_H^{\kappa}$ where each $V_H^j$ is a maximal independent set in $H$.
In any acyclic orientation $P$ of $H$, the source set $S$ satisfies $S \subseteq V_H^j$ for some $j \in \{1,\ldots, \kappa\}$.
Moreover, $V_P(u) = V_P(u')$ for any $u,u' \in S$.
A d.t.d.\ $T$ for $P$ of width $\sw(T) = 1$ is the tree on $|S|$ bags with one source per bag, which can be computed in time $O(\poly(k))$.

Suppose now we add $\epsilon$ arcs to $P$, with any orientation; this means $H$ is a complete multipartite graph plus $\epsilon$ edges.
Again we have $S \subseteq V_H^j$, but now for some $u,u' \in S$ we might have $V_P(u) \ne V_P(u')$, so the d.t.d\ above might not be valid anymore.
Let $P_j = P[V_H^j]$ and consider any d.t.d.\ $T$ for $P_j$.
We argue that $T$ is a valid d.t.d.\ for $P$ as well.
First, the source set of $P_j$ is same of $P$ (that is, $S$).
Thus, since $T$ satisfies properties (1) and (2) of Definition~\ref{def:piecedecomp} for $P_j$, then it does so for $P$, too.
For property (3), note that every node $v \in V_H \setminus V_H^j$ is reachable from every $u \in S$.
Thus, all bags $B$ of $T$ satisfy $V_P(B) = (V_H \setminus V_H^j) \cup (V_P(B) \cap V_H^j)$.
As a consequence, for any three bags $B,B_1,B_2$, if $V_{P_j}(B_1) \cap V_{P_j}(B_2) \subseteq V_{P_j}(B)$ then $V_{P}(B_1) \cap V_{P}(B_2) \subseteq V_{P}(B)$.
Thus $T$ satisfies property (3) and is a d.t.d.\ for $P$.
Therefore, any d.t.d.\ for $P_j$ is a d.t.d.\ for $P$.
Now, since $P_j$ has at most $\epsilon$ edges, by Theorem~\ref{thm:widthbound} in time $2^{O(k)}$ we can compute a d.t.d.\ for it of width at most $\lfloor\frac{\epsilon}{4}\rfloor+2$ on $O(k)$ bags.

Consider now any $\theta \in \Theta(H)$.
For any $v \in V_H$, we denote by $\theta(v)$ the node of $H/\theta$ corresponding to $v$, and for any $x \in H/\theta$ we let $\theta^{-1}(x) = \{u \in V_H : \theta(u)=\theta(x)\}$ be the set of nodes of $H$ identified in $x$.
Let $P$ be any acyclic orientation of $H/\theta$.
Since the sources $S$ of $P$ form an independent set, then $\cup_{x \in S} \theta^{-1}(x) \subseteq V_H^j$ for some $j$.
Moreover, for any node $x$ of $P$, if $v \in \theta^{-1}(x)$ for some $v \in V_H^i$ with $i \ne j$ then $x$ reachable from every node in $S$.
Therefore, if we let $V_P^j = \cup_{v \in V_H^j} \theta(v)$ and $P_j=P[V_P^j]$, the arguments above apply and we obtain the same bound.
\qed
\end{proof}
By coupling Lemma~\ref{lem:hbipartite} and Theorem~\ref{thm:wrapping},
when $H$ is complete multipartite, for computing $\homo{H,G}$ and $\sub{H,G}$ we obtain a time bound of $2^{O(k \log k)} \cdot O(d^{k-1} n \log n)$.
Similarly, when $H$ is complete multipartite plus $\epsilon$ edges, we obtain a time bound of $2^{O(k \log k)} \cdot O(d^{k-\lfloor\frac{\epsilon}{4}\rfloor-2}n^{\lfloor\frac{\epsilon}{4}\rfloor-2}\log n)$.
This proves Theorem~\ref{thm:bipartite}.

\subsection{Independence number and dag treewidth}
\label{sub:alpha}
Recall that $\alpha(H)$ is the independence number of $H$. We show:
\begin{lemma}
\label{thm:pw_tw}
Any $k$-node graph $H$ satisfies $\Omega(\alpha(H)) \le \hsw(H) \le \alpha(H)$.
\end{lemma}
\begin{proof}
For the upper bound, note that $\alpha(H'/\theta) \le \alpha(H)$ for any $H' \in D(H)$ and any $\theta \in \Theta(H')$. Moreover, in any acyclic orientation $P$ of $H'/\theta$ the sources form an independent set. Thus $\sw(P) \le \alpha(H)$. The bound follows by Definition~\ref{def:widths}.

For the lower bound, we exhibit a pattern $H'$ obtained by adding edges to $H$ such that $\sw(P) = \Omega(\alpha(H))$ for all its acyclic orientations $P$.
Let $I \subseteq V_H$ be an independent set of $H$ with $|I| = \Omega(\alpha(H))$ and $|I|\!\!\mod 5 \equiv 0$.
We add edges to $I$, so to obtain the $1$-subdivision of an expander.
Partition $I$ into $I_{\mkey},I_{\roots}$ where $|I_{\mkey}| = \frac{2}{5}|I|$ and $|I_{\roots}| = \frac{3}{5}|I|$.
Consider a $3$-regular expander $\mathcal{E}=(I_{\mkey}, E_{\mathcal{E}})$ of linear treewidth $t(\mathcal{E}) = \Omega(|I_{\mkey}|)$.
It is well known that such expanders exist (see e.g.\ Proposition~1 and Theorem~5 of~\cite{Grohe&2009}).
Note that $|E_{\mathcal{E}}|=\frac{3}{2}|I_{\mkey}|=|I_{\roots}|$.
For each edge $\{u,v\} \in E_{\mathcal{E}}$, we choose a distinct node in $I_{\roots}$, denoted by $e_{uv}$, and we add to $H$ the edges $\{e_{uv},u\}$ and $\{e_{uv},v\}$.
Let $H'$ be the resulting pattern.
Observe that $H'[I]$ is the $1$-subdivision of $\mathcal{E}$, and that $t(\mathcal{E}) = \Omega(\alpha(H))$ since $|I_{\mkey}|=\Omega(\alpha(H))$.

Let now $P=(V_P,A_P)$ be any acyclic orientation of $H'$ where $I_{\roots} \subseteq S_P$ where $S_P$ are the sources of $P$.
Such an orientation exists since $I_{\roots}$ is an independent set in $H'$.
Let $T$ be any d.t.d.\ of $P$.
We show that $\sw(T) \ge \frac{1}{2}(\tw(\mathcal{E})+1) = \Omega(\alpha(H))$, which implies $\sw(P) = \Omega(\alpha(H))$ and therefore the thesis $\hsw(H) = \Omega(\alpha(H))$.
To this end, consider the tree $D$ obtained from $T$ by replacing each bag of sources $B$ with the bag of nodes $J(B) \cap I_{\mkey}$.
We claim that $D$ is a tree decomposition of $\mathcal{E}$ of width at most $2\sw(T)-1$.
Let us start by checking the properties of Definition~\ref{def:treedecomp}.
\\
\textbf{Property (1)}. By point (2) of Definition~\ref{def:piecedecomp}, the d.t.d.\ $T$ satisfies $\cup_{B \in T} B = \roots_P$.
Therefore, by construction of $D$, we have $\cup_{X \in D}\, X = \cup_{B \in T} (J(B) \cap I_{\mkey}) = I_{\mkey}$.
\\
\textbf{Property (2)}. Let $\{v,w\}$ be any edge of $\mathcal{E}$ where we recall that $\{v,w\} \subseteq I_{\mkey}$. By construction of $H'$, there exists $u \in I_{\roots}$ such that $J(u) = \{v,w\}$ in $P$. Since $T$ is a d.t.d., it satisfies point (2) of Definition~\ref{def:piecedecomp}, hence $u \in B$ for some $B \in T$. By construction of $D$ this implies there is some bag $X \in D$ such that $X=J(u)=\{v,w\}$.
\\
\textbf{Property (3)}.
Fix any three bags $X_1,X_2,X_3 \in D$ such that $X_1 \in D(X_2,X_3)$.
By construction, $X_1=J(B_1)\cap I_{\mkey}, X_2=J(B_2)\cap I_{\mkey}, X_3=J(B_3)\cap I_{\mkey}$ for some $B_1,B_2,B_3 \in T$ such that $B_1 \in T(B_2,B_3)$.
Consider any $v \in X_2 \cap X_3$; we need to show that $v \in X_1$.
By construction of $D$, we have $v \in X_2 \cap X_3 = J(B_2) \cap J(B_3) \cap I_{\mkey}$.
Thus, there exist $u \in B_2$ and $u' \in B_3$ such that $v \in J(u) \cap I_{\mkey}$ and $v \in J(u') \cap I_{\mkey}$.
However, since $B_1 \in T(B_2,B_3)$, point (3) of Definition~\ref{def:treedecomp} implies $J(u) \cap J(u') \subseteq J(B_1)$.
Therefore, $v \in J(B_1)$ as well.
Moreover $v \in I_{\mkey}$, and thus $v \in J(B_1) \cap I_{\mkey}$.
But $J(B_1) \cap I_{\mkey} = X_1$, so $v \in X_1$.

Hence, $D$ is a tree decomposition of $\mathcal{E}$.
Finally, note that any bag $X \in D$ by construction satisfies $|X| = |J(B) \cap I_{\mkey}| \le 2|B|$ since any source $u \in B$ has at most $2$ arcs towards $I_{\mkey}$.
Then by Definition~\ref{def:treewidth} and Definition~\ref{def:pw} we have $\tw(\mathcal{E}) \le 2\sw(P)-1$, that is, $\sw(P) \ge \frac{1}{2}(\tw(\mathcal{E})+1)$, as claimed.
\qed
\end{proof}

\section{Lower bounds}
\label{sec:lb}
We prove Theorem~\ref{thm:hsw_lb}, in a more technical form.
Note that, since $\hsw(H) = \Theta(\alpha(H))$ by Lemma~\ref{thm:pw_tw}, the bound still holds if one replaces $\hsw(H)$ by $\alpha(H)$.
The proof uses the following result:
\begin{theorem}[\cite{Curticapean&2014}, Theorem I.2]
\label{thm:curticapean}
The following problems are $\#W[1]$-hard and, assuming ETH, cannot be solved in time $f(k) \cdot n^{o(k/\log{k})}$ for any computable function $f$: counting (directed) paths or cycles of length $k$, and counting edge-colorful or uncolored $k$-matchings in bipartite graphs.
\end{theorem}
Let us now state the lower bound.
\begin{theorem}
Choose any function $a : \mathbb{N} \to \mathbb{N}$ such that $a(k) \in [1,k]$ for all $k \in \mathbb{N}$. There exists an infinite family $\mathcal{H}$ of patterns such that (1) for all $H \in \mathcal{H}$ we have $\hsw(H) = \Theta(a(|V(H)|))$, and (2) if there exists an algorithm that for all $H \in \mathcal{H}$ computes $\ind{H,G}$ or $\sub{H,G}$ in time $f(d,k) \cdot n^{o(a(\tau(H))/\log {a(\tau(H))})}$, where $d$ is the degeneracy of $G$, then ETH fails.
\end{theorem}
\begin{proof}
We reduce counting cycles in an arbitrary graph to counting a gadget pattern on $k$ nodes and dag treewidth $O(a(k))$ in a $d$-degenerate graph.

First, fix a function $d : \mathbb{N} \to \mathbb{N}$ such that $d(k) \in \Omega(\frac{k}{a(k)})$.
Now consider a simple cycle on $k_0 \ge 3$ nodes and any arbitrarily large $k \ge 3$.
Our gadget pattern on $k$ nodes is the following.
For each edge $e=uv$ of the cycle create a clique $C_e$ on $d(k)-1$ nodes; delete $e$ and connect both $u$ and $v$ to every node of $C_e$.
The resulting pattern $H$ has $k =  k_0 \, d(k)$ nodes.
Let us prove that $\hsw(H) \le k_0$; since $k_0 = \frac{k}{d(k)} \in O(a(k))$, this implies $\hsw(H) = O(a(k))$ as desired.
Consider again the generic edge $e=uv$.
In any acyclic orientation $P$ of $H$, the set $C_{e} \cup u$ induces a clique, and thus can contain at most one source.
Applying the argument to all $e$ shows that $|\roots(P)| \le k_0$, hence $\hsw(P) \le k_0$.
This holds also if we add edges and/or identify nodes of $P$, hence $\hsw(H) \le k_0$.

Now consider a simple graph $G_0$ on $n_0$ nodes and $m_0$ edges.
We replace each edge of $G_0$ as described above, which takes time $O(\poly(n_0))$.
The resulting graph $G$ has $n=m_0(d-1) + n_0 = O(dn_0^2)$ nodes and degeneracy $d$.
Every $k_0$-cycle of $G_0$ is univocally associated to a copy of $H$ in $G$ (note that every non-induced copy of $H$ in $G$ is induced too).
Suppose we have an algorithm that computes $\ind{H,G}$ or $\sub{H,G}$ in time $f(d,k) \cdot n^{o(\hsw(H)/\log{\hsw(H)})}$.
Since $\hsw(H) \le k_0$, $n = O(dn_0^2)$, $k=f_1(d,k_0)$, and $d=f_2(k_0)$, for some functions $f_1,f_2$, then the running is time $f(k_0) \cdot (n_0)^{o(k_0/\log{k_0})}$.
Invoking Theorem~\ref{thm:curticapean} concludes the proof.
\qed
\end{proof}

\section{Conclusions}
We have shown how, by introducing a novel tree-like decomposition for directed acyclic graphs, one can improve on the decades-old state-of-the-art subgraph counting algorithms when the host graph is sufficiently sparse.
Our decomposition may be of independent interest, as it seems to capture the relevant structure of the problem.
We leave open the question of finding a characterization of the complexity of subgraph counting in sparse graphs that is tight for every given pattern, rather than pattern classes as a whole.

\section*{Acknowledgements}
I thank the anonymous reviewers for their comments.

\bibliographystyle{plain}
\bibliography{biblio.bib}

\begin{thebibliography}{10}

\bibitem{Alon&2008}
N.~Alon, P.~Dao, I.~Hajirasouliha, F.~Hormozdiari, and S.~C. Sahinalp.
\newblock {Biomolecular network motif counting and discovery by color coding}.
\newblock {\em Bioinformatics}, 24(13):i241--249, July 2008.

\bibitem{Berend&2010}
Daniel Berend and Tamir Tassa.
\newblock Improved bounds on {B}ell numbers and on moments of sums of random
  variables.
\newblock {\em Probability and Math.\ Statistics}, 30(2):185--205, 2010.

\bibitem{Borgs&2006}
Christian Borgs, Jennifer Chayes, L{\'a}szl{\'o} Lov{\'a}sz, Vera~T. S{\'o}s,
  and Katalin Vesztergombi.
\newblock {\em Counting Graph Homomorphisms}, pages 315--371.
\newblock 2006.

\bibitem{Chen&2005}
Jianer Chen, Benny Chor, Mike Fellows, Xiuzhen Huang, David Juedes, Iyad~A.
  Kanj, and Ge~Xia.
\newblock Tight lower bounds for certain parameterized {NP}-hard problems.
\newblock {\em Information and Computation}, 201(2):216 -- 231, 2005.

\bibitem{Chen&2006}
Jianer Chen, Xiuzhen Huang, Iyad~A. Kanj, and Ge~Xia.
\newblock Strong computational lower bounds via parameterized complexity.
\newblock {\em Journal of Computer and System Sciences}, 72(8):1346 -- 1367,
  2006.

\bibitem{Chiba&1985}
Norishige Chiba and Takao Nishizeki.
\newblock Arboricity and subgraph listing algorithms.
\newblock {\em SIAM J.\ Comput.}, 14(1):210--223, February 1985.

\bibitem{Curticapean&2017}
Radu Curticapean, Holger Dell, and D{\'{a}}niel Marx.
\newblock Homomorphisms are a good basis for counting small subgraphs.
\newblock In {\em Proc.\ of ACM STOC}, pages 210--223, 2017.

\bibitem{Curticapean&2014}
Radu Curticapean and D\'{a}niel Marx.
\newblock Complexity of counting subgraphs: Only the boundedness of the
  vertex-cover number counts.
\newblock In {\em Proc.\ of IEEE FOCS}, pages 130--139, 2014.

\bibitem{Diestel2017}
Reinhard Diestel.
\newblock {\em Graph Theory}.
\newblock Springer Publishing Company, Incorporated, 5th edition, 2017.

\bibitem{Eppstein1994}
David Eppstein.
\newblock Arboricity and bipartite subgraph listing algorithms.
\newblock {\em Inf. Process. Lett.}, 51(4):207--211, 1994.

\bibitem{Eppstein1995}
David Eppstein.
\newblock Subgraph isomorphism in planar graphs and related problems.
\newblock {\em Journal of Graph Algorithms and Applications}, 3(3):1--27, 1999.

\bibitem{Eppstein&2010}
David Eppstein, Maarten Löffler, and Darren Strash.
\newblock Listing all maximal cliques in sparse graphs in near-optimal time.
\newblock In {\em Algorithms and Computation}, pages 403--414. Springer Berlin
  Heidelberg, 2010.

\bibitem{Eppstein&2011}
David Eppstein and Darren Strash.
\newblock Listing all maximal cliques in large sparse real-world graphs.
\newblock In {\em Experimental Algorithms}, pages 364--375. Springer Berlin
  Heidelberg, 2011.

\bibitem{Fomin&2010}
Fedor~V. Fomin and Dieter Kratsch.
\newblock {\em Exact Exponential Algorithms}.
\newblock Springer-Verlag Berlin Heidelberg, 1 edition, 2010.

\bibitem{Ganian2010digraph}
Robert Ganian, Petr Hlin{\v{e}}n{\'{y}}, Joachim Kneis, Daniel Meister, Jan
  Obdr{\v{z}}{\'{a}}lek, Peter Rossmanith, and Somnath Sikdar.
\newblock Are there any good digraph width measures?
\newblock {\em Journal of Combinatorial Theory, Series B}, 116:250--286,
  January 2016.

\bibitem{Grohe&2013nowheredense}
Martin Grohe, Stephan Kreutzer, and Sebastian Siebertz.
\newblock {Characterisations of Nowhere Dense Graphs (Invited Talk)}.
\newblock In {\em Proc.\ of FSTTCS}, volume~24, pages 21--40, 2013.

\bibitem{Grohe&2009}
Martin Grohe and D\'{a}niel Marx.
\newblock On tree width, bramble size, and expansion.
\newblock {\em Journal of Combinatorial Theory, Series B}, 99(1):218 -- 228,
  2009.

\bibitem{Grohe&FO}
Martin Grohe and Nicole Schweikardt.
\newblock First-order query evaluation with cardinality conditions.
\newblock In {\em Proc.\ of ACM SIGMOD}, page 253–266, 2018.

\bibitem{Impagliazzo&1998}
Russell Impagliazzo, Ramamohan Paturi, and Francis Zane.
\newblock Which problems have strongly exponential complexity?
\newblock In {\em Proc.\ of IEEE FOCS}, pages 653--662, 1998.

\bibitem{Kneis&2005}
Joachim Kneis, Daniel M\"olle, Stefan Richter, and Peter Rossmanith.
\newblock Algorithms based on the treewidth of sparse graphs.
\newblock In {\em Graph-Theoretic Concepts in Computer Science}, pages
  385--396. Springer Berlin Heidelberg, 2005.

\bibitem{LeGall2014}
Fran\c{c}ois Le~Gall.
\newblock Powers of tensors and fast matrix multiplication.
\newblock In {\em Proc.\ of ISSAC}, pages 296--303, 2014.

\bibitem{Mathon1979}
Rudolf Mathon.
\newblock A note on the graph isomorphism counting problem.
\newblock {\em Information Processing Letters}, 8(3):131 -- 136, 1979.

\bibitem{Nesetril2012sparsity}
J.~Ne{\v{s}}et{\v{r}}il and P.O. de~Mendez.
\newblock {\em Sparsity: Graphs, Structures, and Algorithms}.
\newblock Algorithms and Combinatorics. Springer Berlin Heidelberg, 2012.

\bibitem{nowhere-dense}
Jaroslav Ne\v{s}et\v{r}il and Patrice {Ossona de Mendez}.
\newblock On nowhere dense graphs.
\newblock {\em European Journal of Combinatorics}, 32(4):600 -- 617, 2011.

\bibitem{Nesetril&1985}
Jaroslav Ne\v{s}et\v{r}il and Svatopluk Poljak.
\newblock On the complexity of the subgraph problem.
\newblock {\em Commentationes Mathematicae Universitatis Carolinae},
  026(2):415--419, 1985.

\bibitem{Patel&2018}
Viresh Patel and Guus Regts.
\newblock Computing the number of induced copies of a fixed graph in a bounded
  degree graph.
\newblock {\em Algorithmica}, 81(5):1844--1858, September 2018.

\bibitem{Sariyuce&2018}
Ahmet~Erdem Sariy\"{u}ce and Ali Pinar.
\newblock Peeling bipartite networks for dense subgraph discovery.
\newblock In {\em Proc.\ of ACM WSDM}, pages 504--512, 2018.

\bibitem{Sariyuce&2017}
Ahmet~Erdem Sariy\"{u}ce, C.~Seshadhri, Ali Pinar, and \"{U}mit~V.
  \c{C}ataly\"{u}rek.
\newblock Nucleus decompositions for identifying hierarchy of dense subgraphs.
\newblock {\em ACM Trans. Web}, 11(3):16:1--16:27, 2017.

\bibitem{Tsourakakis&2017}
Charalampos~E. Tsourakakis, Jakub Pachocki, and Michael Mitzenmacher.
\newblock Scalable motif-aware graph clustering.
\newblock In {\em Proc.\ of WWW}, pages 1451--1460, 2017.

\end{thebibliography}

\end{document}